\documentclass[letter,journal]{ieeetran}  

\usepackage{color,soul}

\usepackage{amssymb,upgreek,nicefrac,amsmath}
\usepackage[]{graphicx}

\renewcommand{\baselinestretch}{1.1}
\newcommand{\Cc}{\mathcal{C}_b}

\newcommand{\cA}{{\mathcal{A}}}
\newcommand{\cX}{{\mathcal{X}}}
\newcommand{\cZ}{{\mathcal{Z}}}
\newcommand{\cS}{{\mathcal{S}}}

\newcommand{\Bor}{{\mathfrak{B}}}
\newcommand{\cI}{\mathcal{I}}
\newcommand{\df}{\doteq}
\newcommand{\reward}[2]{u_{#1}(#2)}
\newcommand{\Prob}{{\mathbb{P}}} 
\newcommand{\Exp}{{\mathbb{E}}} 

\newcommand{\supnorm}[1]{{\lVert}#1{\rVert}_{\infty}}

\newcommand{\Neig}{\mathcal{N}}

\newcommand{\Dirac}[1]{\boldsymbol{\delta}_{#1}}
\newtheorem{definition}{Definition}[section]

\newtheorem{lemma}{Lemma}[section]
\newtheorem{theorem}{Theorem}[section]
\newtheorem{proposition}{Proposition}[section]
\newtheorem{property}{Property}[section]

\newenvironment{proof}{\textbf{Proof.}}{$\square$\\}

\newcommand{\tr}{^{\mathrm T}}

\newcommand{\magn}[1]{\left\vert #1 \right\vert}

\begin{document}

\title{Stochastic Stability of Perturbed Learning Automata in Positive-Utility Games\thanks{An earlier version of parts of this paper appeared in \cite{Chasparis17_PLA}. This work has been partially supported by the European Union grant EU H2020-ICT-2014-1 project RePhrase (No. 644235). It has also been partially supported by the Austrian Ministry for Transport, Innovation and Technology, the Federal Ministry of Science, Research and Economy, and the Province of Upper Austria in the frame of the COMET center SCCH.}}

\author{Georgios C. Chasparis\thanks{G. C. Chasparis is with the Department of Data Analysis Systems, Software Competence Center Hagenberg GmbH, Softwarepark 21, A-4232 Hagenberg, Austria, E-mail: georgios.chasparis@scch.at.}}

\maketitle

\begin{abstract}
This paper considers a class of reinforcement-based learning (namely, \emph{perturbed learning automata}) and provides a stochastic-stability analysis in repeatedly-played, positive-utility, finite strategic-form games. Prior work in this class of learning dynamics primarily analyzes asymptotic convergence through stochastic approximations, where convergence can be associated with the limit points of an ordinary-differential equation (ODE). However, analyzing global convergence through an ODE-approximation requires the existence of a Lyapunov or a potential function, which naturally restricts the analysis to a fine class of games. To overcome these limitations, this paper introduces an alternative framework for analyzing asymptotic convergence that is based upon an explicit characterization of the invariant probability measure of the induced Markov chain. We further provide a methodology for computing the invariant probability measure in positive-utility games, together with an illustration in the context of coordination games.  
\end{abstract}

\section{Introduction} \label{sec:Introduction}

Recently, multi-agent formulations have been utilized to tackle distributed optimization problems, since communication and computational complexity might be an issue under centralized schemes. In such formulations, decisions are usually taken in a \emph{repeated} fashion, where agents select their next actions based on their \emph{own} prior experience. In the case of finite number of actions for each agent, such multi-agent interactions can be designed as strategic-form games, where agents are repeatedly involved in a strategic interaction with a fixed \emph{payoff} or \emph{utility} function. Such framework finds numerous applications, including, for example, the problem of distributed overlay routing \cite{Chun04}, distributed topology control \cite{Komali08} and distributed resource allocation \cite{Wei10}.

Given the repeated fashion of the involved strategic interactions in such formulations, several questions naturally emerge, including: a) \emph{Can agents ``learn'' to asymptotically select optimal actions?}, b) \emph{What information should agents share with each other?}, and c) \emph{What is the computational complexity of the learning process?} Under the scope of engineering applications, it is usually desirable that each agent shares minimum amount of information with other agents, while the computational complexity of the learning process is small. Naturally, \emph{payoff-based learning} has drawn significant attention. Under such class of learning dynamics, each agent receives \emph{only} measurements of its own utility function, while the details of this function (i.e., its mathematical formula) are unknown. Furthermore, each agent cannot access the actions selected or utilities received by other agents.


In such repeatedly-played strategic-form games, a popular objective for payoff-based learning is to guarantee convergence (in some sense) to Nash equilibria. Convergence to Nash equilibria may be desirable, especially when the set of optimal centralized solutions belongs to the set of Nash equilibria. 

\emph{Reinforcement-based learning} has been utilized in strategic-form games in order for agents to gradually learn to play Nash equilibria. It may appear under alternative forms, including discrete-time replicator dynamics \cite{Arthur93}, learning automata \cite{Tsetlin73,Narendra89} and $Q$-learning \cite{hu_nash_2003}. In all these classes of learning dynamics, deriving conditions under which convergence to Nash equilibria is achieved may not be a trivial task especially in the case of large number of agents (as it will be discussed in detail in the forthcoming Section~\ref{sec:PerturbedLearningAutomata}).

In the present paper, we consider a class of reinforcement-based learning introduced in \cite{ChasparisShamma11_DGA} that is closely related to both discrete-time replicator dynamics and learning automata. We will refer to this class of dynamics as \emph{perturbed learning automata}. The main difference with prior reinforcement-based learning schemes lies in a) the step-size sequence, and b) the perturbation (or \emph{mutations}) term. The step-size sequence is assumed constant, thus introducing a fading-memory effect of past experiences in each agent's strategy. On the other hand, the perturbation term introduces errors in the selection process of each agent. Both these two features can be used for designing a desirable asymptotic behavior.

We provide an analytical framework for deriving conclusions over the asymptotic behavior of the dynamics that is based upon an explicit characterization of the invariant probability measure of the induced Markov chain. In particular, we show that in all finite strategic-form games satisfying the \emph{Positive-Utility Property} (i.e., games with strictly positive utilities), the support of the invariant probability measure coincides with the set of pure strategy profiles. Furthermore, we provide a methodology for computing the set of stochastically stable states in all positive-utility games. We illustrate this methodology in the context of coordination games and provide a simulation study in distributed network formation. This illustration is also of independent interest since it extends prior work in coordination games under reinforcement-based learning, where convergence to mixed strategy profiles may only be excluded under strong conditions in the utility function (e.g., existence of a potential function).

In the remainder of the paper, Section~\ref{sec:PerturbedLearningAutomata} presents the investigated class of learning dynamics, related work and the main contributions. Section~\ref{sec:StochasticStability} provides a simplification in the characterization of stochastic stability, while Section~\ref{sec:TechnicalDerivation} presents its technical derivation. This result is utilized for computing the stochastically stable states in positive-utility games in Section~\ref{sec:StochasticallyStableStates}. In Section~\ref{sec:Illustration}, we present an illustration of the proposed methodology in the context of coordination games, together with a simulation study in distributed network formation. Finally, Section~\ref{sec:Conclusions} presents concluding remarks.

{\bf Notation:}
\begin{itemize}
\item For a Euclidean topological space $\cZ\subset\mathbb{R}^{n}$, let $\Neig_{\delta}(x)$ denote the $\delta$-neighborhood of $x\in\cZ$, i.e.,
$\Neig_{\delta}(x) \df \{y\in\cZ:|x-y|<\delta\},$
where $|\cdot|$ denotes the Euclidean distance.
\item $e_j$ denotes the \emph{unit vector} in $\mathbb{R}^{n}$ where its $j$th entry is equal to 1 and all other entries are equal to 0.
\item $\Delta(n)$ denotes the \emph{probability simplex} of dimension $n$, i.e.,
$\Delta(n) \df \left\{ x\in\mathbb{R}^{n} : x\geq{0}, \mathbf{1}\tr x=1 \right\}.$
\item For some set $A$ in a topological space $\cZ$, let $\mathbb{I}_{A}:\cZ\to\{0,1\}$ denote the index function, i.e.,
\begin{eqnarray*}
\mathbb{I}_{A}(x) \df \begin{cases}
1 & \mbox{ if } x\in{A}, \\
0 & \mbox{ else.}
\end{cases}
\end{eqnarray*}

\item For a finite set $A$, $\magn{A}$ denotes its cardinality.
\item For a finite set $A$ and any probability distribution $\sigma\in\Delta(\magn{A})$, the random selection of an element of $A$ will be denoted by ${\rm rand}_{\sigma}[A]$. If $\sigma=(\nicefrac{1}{\magn{A}},...,\nicefrac{1}{\magn{A}})$, the random selection will be denoted by ${\rm rand}_{\rm unif}[A]$.

\item $\Dirac{x}$ denotes the Dirac measure at $x$.

\item $\log(\cdot)$ denotes the natural logarithm.

\end{itemize}

\section{Perturbed Learning Automata}	\label{sec:PerturbedLearningAutomata}

\subsection{Terminology}

We consider the standard setup of \emph{finite strategic-form games}. Consider a finite set of \emph{agents} (or \emph{players}) $\mathcal{I} = \{1,...,n\}$, and let each agent $i$ have a finite set of actions $\mathcal{A}_i$. Let $\alpha_i\in\mathcal{A}_i$ denote any such action of agent $i$. The set of \emph{action profiles} is the Cartesian product $\mathcal{A}\df\mathcal{A}_1\times\cdots\times\mathcal{A}_n$ and let $\alpha=(\alpha_1,...,\alpha_n)$ be a representative element of this set. We will denote $-i$ to be the complementary set $\cI\backslash{i}$ and often decompose an action profile as follows $\alpha=(\alpha_i,\alpha_{-i})$. The \emph{payoff/utility function} of agent $i$ is a mapping $\reward{i}{\cdot}:\mathcal{A}\to\mathbb{R}$. A \emph{finite strategic-form game} is defined by the triple $\langle{\mathcal{I},\mathcal{A},\{\reward{i}{\cdot}\}_i}\rangle$. 

\textit{\textbf{For the remainder of the paper}}, we will be concerned with finite strategic-form games that satisfy the \emph{\textbf{Positive-Utility Property}}. 

\begin{property}[Positive-Utility Property]		\label{P:PositiveUtilityProperty}
For any agent $i\in\mathcal{I}$ and any action profile $\alpha\in\mathcal{A}$, $\reward{i}{\alpha}>0$.
\end{property}

This property is rather generic and applies to a large family of games. For example, games at which some form of alignment of interests exists between agents (e.g., \emph{coordination games} \cite{ChasparisAriShamma13_SIAM} or \emph{weakly-acyclic games} \cite{marden_payoff_2009}), can be designed to satisfy this property, since agents' utilities/preferences are rather close to each other at any given action profile. However, in the forthcoming analysis, we do not impose any structural constraint but Property~\ref{P:PositiveUtilityProperty}.

\subsection{Perturbed Learning Automata}

We consider a form of reinforcement-based learning that belongs to the general class of \emph{learning automata} \cite{Narendra89}. In learning automata, each agent updates a finite probability distribution $x_i\in\cX_i\df\Delta(\magn{\cA_i})$ representing its beliefs about the most profitable action. 

The proposed learning model is described in Table~\ref{Tb:ReinforcementLearning}. At the first step, each agent $i$ updates its action given its current strategy vector $x_i(t)$. Its selection is slightly perturbed by a perturbation (or \emph{mutations}) factor $\lambda>0$, such that, with a small probability $\lambda$ agent $i$ follows a uniform strategy (or, it \emph{trembles}).  At the second step, agent $i$ evaluates its new selection by collecting a utility measurement, while in the last step, agent $i$ updates its strategy vector.

\begin{table}[t!]
\caption{Perturbed Learning Automata}
\boxed{\small
\begin{minipage}{0.48\textwidth}
At fixed time instances $t=1,2,...$, and for each agent $i\in\cI$, the following steps are executed recursively. Let $\alpha_i(t)$ and $x_i(t)$ denote the current action and strategy of agent $i$, respectively. 
\begin{enumerate}
\item (\emph{\textbf{action update}}) Agent $i$ selects a new action $\alpha_i(t+1)$ as follows: 
\begin{equation}	\label{eq:ActionUpdate}
\alpha_i(t+1) = \begin{cases}
{\rm rand}_{x_i(t)}[\mathcal{A}_i], & \mbox{ with probability } 1-\lambda, \\
{\rm rand}_{\rm unif}[\mathcal{A}_i], & \mbox{ with probability } \lambda,
\end{cases}
\end{equation} 
for some small perturbation factor $\lambda>0$.

\item (\emph{\textbf{evaluation}}) Agent $i$ applies its new action $\alpha_i(t+1)$ and receives a measurement of its utility $\reward{i}{\alpha(t+1)}>0$. 

\item (\emph{\textbf{strategy update}}) Agent $i$ revises its strategy $x_i\in\Delta(\magn{\mathcal{A}_i})$ as follows: 
\begin{eqnarray}	\label{eq:ReinforcementLearningModel}
\lefteqn{x_i(t+1)} \cr & = & x_i(t) + \epsilon \cdot \reward{i}{\alpha(t+1)} \cdot [e_{\alpha_i(t+1)} - x_i(t)] \cr & \df & \mathcal{R}_{i}(\alpha(t+1),x_i(t)),
\end{eqnarray}
for some constant step-size $\epsilon>0$.
\end{enumerate}
\end{minipage}
}
\label{Tb:ReinforcementLearning}
\end{table}

Here, we identify actions $\mathcal{A}_i$ with vertices of the simplex, $\{e_1,...,e_{\magn{\mathcal{A}_i}}\}$. For example, if agent $i$ selects its $j$th action at time $t$, then $e_{\alpha_i(t)}\equiv e_j$. To better see how the strategies evolve, let us consider the following toy example. Let the current strategy of agent $i$ be $x_i(t) = \left(\begin{array}{cc} \nicefrac{1}{2} & \nicefrac{1}{2} \end{array}\right)\tr$, i.e., agent $i$ has two actions, each assigned probability $\nicefrac{1}{2}$. Let also $\alpha_i(t+1)=1$, i.e., agent $i$ selects the first action according to rule (\ref{eq:ActionUpdate}). Then, the new strategy vector for agent $i$ is:
\begin{equation*}
x_i(t+1) = \nicefrac{1}{2} \left(\begin{array}{c} 1 + \epsilon u_i(\alpha(t+1))\\ 1 - \epsilon u_i(\alpha(t+1)) \end{array}\right).
\end{equation*}
Note that the strategy of the selected action increased by an amount that is proportional to the reward received. In other words, the dynamics reinforce repeated selection, and the reinforcement size, $\epsilon u_i(\alpha(t+1))$, depends on the reward received.

By playing a strategic-form game repeatedly over time, players do not always experience the same reward when selecting the same action, since other players may also change their actions. This dynamic element of the reinforcement size is the factor that complicates its convergence analysis, as it will become clear in the forthcoming Section~\ref{sec:RelatedWork}.

Note that by letting the step-size $\epsilon$ to be sufficiently small and since the utility function $\reward{i}{\cdot}$ is uniformly bounded in $\mathcal{A}$, $x_i(t)\in\Delta(\magn{\mathcal{A}_i})$ for all $t$. 

In case $\lambda=0$, the above update recursion will be referred to as the \emph{unperturbed learning automata}.

\subsection{Related work}	\label{sec:RelatedWork}

In this section, we provide a short overview of alternative payoff-based learning schemes specifically designed for repeatedly-played strategic-form games with a \emph{finite} set of actions and a \emph{fixed} utility function for each player. We have identified four main classes of payoff-based dynamics under such structural assumptions, namely \emph{discrete-time replicator dynamics}, \emph{learning automata}, \emph{$Q$-learning}, and \emph{aspiration-based learning}. Note that payoff-based learning has also been applied to static games with continuous action sets, e.g., extremum-seeking control \cite{frihauf_nash_2012,ye_distributed_2015} or actor-critic reinforcement learning \cite{perkins_mixed-strategy_2017}. The focus here instead is only on \emph{finite} action sets.

\subsubsection*{Discrete-time replicator dynamics}

A type of learning dynamics which is quite closely related to the dynamics of Table~\ref{Tb:ReinforcementLearning} is the discrete-time version of \emph{replicator dynamics} (cf.,~\cite{Hofbauer98}). There have been several variations with respect to the selection of the step-size sequence. For example, Arthur \cite{Arthur93} considered a similar rule, with $\lambda=0$ and step-size $\epsilon_i(t) = 1/(ct^{\nu}+\reward{i}{\alpha(t+1)}$, for some positive constant $c$ and for $\nu\in(0,1)$ (in the place of the constant step-size $\epsilon$ of (\ref{eq:ReinforcementLearningModel})). A comparative model is also used by Hopkins and Posch in \cite{HopkinsPosch05}, with $\epsilon_i(t) = 1/(V_i(t)+\reward{i}{\alpha(t+1)})$, where $V_i(t)$ is the accumulated benefits of agent $i$ up to time $t$, which gives rise to the urn process of Erev-Roth \cite{Erev98}. Some similarities are also shared with the Cross' learning model of \cite{BorgersSarin97}, where $\epsilon_i(t)=1$ and $\reward{i}{\alpha(t)}\leq{1}$, and its modification presented by Leslie in \cite{Leslie04}, where $\epsilon(t)$, instead, is decreasing with time.

The main difference of the proposed dynamics of Table~\ref{Tb:ReinforcementLearning} lies in the perturbation parameter $\lambda>0$ which was first introduced and analyzed in \cite{ChasparisShamma11_DGA}. A state-dependent perturbation term has also been investigated in \cite{ChasparisShammaRantzer15}. The perturbation parameter may serve as an equilibrium selection mechanism, since it may exclude convergence to \emph{action profiles that are not Nash equilibria} (briefly, \emph{non-Nash action profiles}) \cite{ChasparisShamma11_DGA}. It resolved one of the main issues of discrete-time replicator dynamics, that is the positive probability of convergence to non-Nash action profiles.

Although excluding convergence to non-Nash action profiles can be guaranteed by sufficiently small $\lambda>0$, establishing convergence to action profiles that are Nash equilibria (\emph{pure Nash equilibria}) may still be an issue. This is desirable in the context of coordination games \cite{Lewis02}, where Pareto-efficient outcomes are usually pure Nash equilibria (see, e.g., the definition of a coordination game in \cite{ChasparisAriShamma13_SIAM}). As shown in \cite{ChasparisShammaRantzer15}, convergence to pure Nash equilibria can be guaranteed only under strong conditions in the utility function. For example, as shown in \cite[Proposition~8]{ChasparisShammaRantzer15}, and under the ODE-method for stochastic approximations, it requires a) the existence of a potential function, and b) conditions over the Jacobian matrix of the potential function. Even if a potential function does exist, verifying conditions (b) is practically infeasible for games of more than 2 players \cite{ChasparisShammaRantzer15}.

On the other hand, an important side-benefit of using this class of dynamics is the indirect ``filtering'' of the utility-function measurements (through the formulation of the strategy vectors in (\ref{eq:ReinforcementLearningModel})). This is demonstrated, for example, in \cite{HopkinsPosch05} for the Erev-Roth model \cite{Erev98}, where the robustness of convergence/non-convergence asymptotic results is presented under the presence of noise in the utility measurements.

\subsubsection*{Learning automata}

Learning automata, as first introduced by \cite{Tsetlin73}, have been used to the control of complex systems due to their simple structure and low computational complexity (cf.,~\cite[Chapter~1]{Narendra89}). \emph{Variable-structure stochastic automata} may incorporate a form of reinforcement of favorable actions, similarly to the replicator dynamics discussed above. An example is the \emph{linear reward-inaction scheme} \cite[Chapter~4]{Narendra89}. Comparing it with the reinforcement rule of (\ref{eq:ReinforcementLearningModel}), the linear reward-inaction scheme accepts a utility of the form $u_i(\alpha)\in\{0,1\}$, where $0$ corresponds to an unfavorable  response and $1$ corresponds to a favorable one. More general forms can also be used when the utility function may accept discrete or continuous values in the unit interval $[0,1]$. 

Analysis of learning automata in games has been restricted to zero-sum and identical-interest games \cite{Narendra89,Sastry94}. In identical interest games, convergence analysis has been derived for small number of players and actions, due to the difficulty in deriving conditions for \emph{absolute monotonicity}, which corresponds to the property that \emph{the expected utility received by each player increases monotonically in time} (cf.,~\cite[Definition~8.1]{Narendra89}). Similar are the results presented in \cite{Sastry94}.

The property of \emph{absolute monotonicity} is closely related to the existence of a \emph{potential function}, as in the case of potential games \cite{MondererShapley96}. Similarly to the discrete-time replicator dynamics, convergence to non-Nash action profiles cannot be excluded when the step-size sequence is constant, even if the utility function satisfies $u_i(\alpha)\in[0,1]$. (The behavior under decreasing step-size is different as \cite[Proposition~2]{ChasparisShammaRantzer15} has shown.) Furthermore, deriving conditions for excluding convergence to mixed strategy profiles in coordination games continues to be an issue, as in discrete-time replicator dynamics. 

Recognizing these issues, reference~\cite{verbeeck_exploring_2007} introduced a class of linear reward-inaction schemes in combination with a coordinated exploration phase so that convergence to the efficient (pure) Nash equilibrium is achieved. However, coordination of the exploration phase requires communication between the players, an approach that does not fit to the distributed nature of dynamics pursued here.

\subsubsection*{$Q$-learning}

Similar questions of convergence to Nash equilibria also appear in alternative reinforcement-based learning formulations, such as approximate dynamic programming and $Q$-learning. Usually, under $Q$-learning, players keep track of the discounted running average reward received by each action, based on which optimal decisions are made (see, e.g., \cite{leslie_individual_2005}). Convergence to Nash equilibria can be accomplished under a stronger set of assumptions, which increases the computational complexity of the dynamics. For example, in the Nash-Q learning algorithm of \cite{hu_nash_2003}, it is indirectly assumed that agents need to have full access to the joint action space and the rewards received by other agents. 

More recently, reference \cite{chapman_convergent_2013} introduced a $Q$-learning scheme in combination with either adaptive play or better-reply dynamics in order to attain convergence to Nash equilibria in potential games \cite{MondererShapley96} or weakly-acyclic games. However, this form of dynamics requires that each player observes the actions selected by the other players, since a $Q$-value needs to be assigned to each joint action.

When the evaluation of the $Q$-values is totally independent, as in the individual $Q$-learning in \cite{leslie_individual_2005}, then convergence to Nash equilibria has been shown only for 2-player zero-sum games and 2-player partnership games with countably many Nash equilibria. Currently, there exist no convergence results in multi-player games. To overcome this deficiency of $Q$-learning, in the context of stochastic dynamic games, reference \cite{arslan_decentralized_2016} employs an additional feature (motivated by \cite{marden_payoff_2009}), namely \emph{exploration phases}. In any such \emph{exploration phase}, \emph{all} agents use constant policies, something that allows for an accurate computation of the optimal $Q$-factors. We may argue that the introduction of common exploration phases for all agents partially destroys the distributed nature of the dynamics, since it requires synchronization between agents.

\subsubsection*{Aspiration-based learning}

Recently, there have been several attempts to establish convergence to Nash equilibria through alternative payoff-based learning dynamics, e.g., the \emph{benchmark-based dynamics} of \cite{marden_payoff_2009} for convergence to Nash equilibria in weakly-acyclic games, the \emph{trial-and-error learning} \cite{young_learning_2009} for convergence to Nash equilibria in generic games, the \emph{mood-based dynamics} of \cite{marden_pareto_2014} for maximizing welfare in generic games and the \emph{aspiration learning} in \cite{ChasparisAriShamma13_SIAM} for convergence to efficient outcomes in coordination games. We will refer to such approaches as \emph{aspiration-based learning}. For these types of dynamics, convergence to Nash equilibria or efficient outcomes can be established without requiring any strong monotonicity properties (as in the multi-player weakly-acyclic games in \cite{marden_payoff_2009}).

The case of noisy utility measurements, which are present in many engineering applications, has not currently been addressed through aspiration-based learning. The only exception is reference \cite{marden_payoff_2009}, under benchmark-based dynamics, where (synchronized) \emph{exploration phases} are introduced through which each agent plays a fixed action for the duration of the exploration phase. If such exploration phases are large in duration (as required by the results in \cite{marden_payoff_2009}), this may reduce the robustness of the dynamics to dynamic changes in the environment (e.g., changes in the utility function). One reason that such robustness analysis is currently not possible in this class of dynamics is the fact that decisions are taken directly based on the measured performances (e.g., by comparing the currently measured performance with the benchmark performance in \cite{marden_payoff_2009}).

\subsection{Contributions}

The aforementioned literature in payoff-based learning dynamics in finite strategic-form games can be grouped into two main categories, namely \emph{reinforcement-based learning} (including discrete-time replicator dynamics, learning automata and $Q$-learning) and \emph{aspiration-based learning}. Summarizing their main advantages/disadvantages, we may argue the following high-level observations. 
\begin{itemize}
\item[(O1)] \emph{Strong asymptotic convergence guarantees} for large number of players, even for generic games, are currently possible under aspiration-based learning. Similar results in reinforcement-based learning are currently restricted to games of small number of players and under strong structural assumptions (e.g., the existence of a potential function). See, for example, the discussion on discrete-time replicator dynamics or learning automata in \cite{ChasparisShammaRantzer15}, or the discussion on $Q$-learning in \cite{arslan_decentralized_2016}.
\item[(O2)] \emph{Noisy observations} can be ``handled'' through reinforcement-based learning due to the indirect \emph{filtering} of the observation signals (e.g., through the formulation of the strategy-vector in the dynamics of Table~\ref{Tb:ReinforcementLearning} or through the formulation of the $Q$ factors in $Q$-learning). This is demonstrated, for example, in the  convergence/non-convergence asymptotic results presented in \cite{HopkinsPosch05} for a variation of the proposed learning dynamics of Table~\ref{Tb:ReinforcementLearning} (with $\lambda=0$ and decreasing $\epsilon$) and under the presence of noise. Similar effects in aspiration-based learning can currently be achieved only through the introduction of \emph{synchronized exploration phases}, as discussed in Section~\ref{sec:RelatedWork}.
\end{itemize}

Motivated by these two observations (O1)--(O2), and the obvious inability of reinforcement-based learning to provide strong asymptotic convergence guarantees in large games, this paper advances the asymptotic convergence guarantees for a class of reinforcement-based learning described in Table~\ref{Tb:ReinforcementLearning}. Our goal is to go beyond common restrictions of small number of players and strong assumptions in the game structure (such as the existence of a potential function).

The proposed dynamics (also \emph{perturbed learning automata}) were first introduced in \cite{ChasparisShamma11_DGA} to resolve stability issues in the boundary of the domain appearing in prior schemes \cite{Arthur93,HopkinsPosch05}. This was achieved through the introduction of the perturbation factor $\lambda$  of Table~\ref{Tb:ReinforcementLearning}. However, strong convergence guarantees (e.g., w.p.1 convergence to Nash equilibria or efficient outcomes) is currently limited to small number of players and under strict structural assumptions, e.g., the existence of a potential function and conditions on its Jacobian matrix \cite{ChasparisShammaRantzer15}.

In this paper, we drop the assumption of a decreasing step-size sequence, and instead we consider the case of a \emph{constant} step-size $\epsilon>0$. Such selection increases the adaptivity of the dynamics to varying conditions (e.g., the number of agents or the utility function). Furthermore, we provide a stochastic-stability analysis that provides a detailed characterization of the invariant probability measure of the induced Markov chain. In particular, our contributions are as follows:
\begin{enumerate}
\item[(C1)] We provide an equivalent finite-dimensional characterization of the infinite-dimensional induced Markov chain of the dynamics, that simplifies significantly the computation of its invariant probability measure. This simplification is based upon a weak-convergence result and it applies to any finite strategic-form game with the Positive-Utility Property \ref{P:PositiveUtilityProperty} (\emph{Theorem~\ref{Th:StochasticStability}}).
\item[(C2)] We capitalize on this simplification and provide a methodology for computing stochastically stable states in positive-utility finite strategic-form games (\emph{Theorem~\ref{Th:StochasticallyStableStatesMinimumResistance}}). 
\item[(C3)] We illustrate the utility of this methodology in establishing stochastic stability in a class of coordination games with no restriction on the number of players or actions (\emph{Theorem~\ref{Th:StochasticStabilityCoordinationGames}}).
\end{enumerate}
These contributions significantly extend the utility of reinforcement-based learning given observation (O1). Note that (C2) does not impose any structural assumptions other than the positive-utility property. Furthermore, (C3) is of independent interest. To the best of our knowledge, (C3) is the first convergence result in the context of reinforcement-based learning in repeatedly-played finite strategic-form games with the following features: a) a completely distributed setup (i.e., without any information exchange between players), b) more than two players, and c) a weakly-acyclicity condition that does not require the existence of a potential function.

The derived convergence results may not be as strong as the ones currently derived under aspiration-based learning, as discussed in Section~\ref{sec:RelatedWork}. However, reinforcement-based learning may better incorporate noisy observations (as discussed in observation (O2)). Moreover, additional features may allow for stronger convergence guarantees, even to Pareto-efficient outcomes, as presented in \cite{ChasparisShamma11_DGA}. Given the simplified analytical framework presented here, the prospects of even stronger convergence guarantees are promising.

This paper is an extension over an earlier version appeared in \cite{Chasparis17_PLA}, which only focused on contribution (C1) above.

\section{Stochastic Stability}	\label{sec:StochasticStability}

In this section, we provide a characterization of the \emph{invariant probability measure} $\mu_{\lambda}$ of the induced Markov chain $P_{\lambda}$ of the dynamics of Table~\ref{Tb:ReinforcementLearning}. The importance lies in an equivalence relation (established through a weak-convergence argument) of $\mu_{\lambda}$ with an invariant distribution of a finite-state Markov chain. Characterization of the stochastic stability of the dynamics will follow directly due to the Birkhoff's individual ergodic theorem.
This simplification in the characterization of $\mu_{\lambda}$ will be the first important step for providing specialized results for stochastic stability in strategic-form games.

\subsection{Terminology and notation}	\label{sec:Terminology}

Let $\cZ\df \mathcal{A}\times \cX$, where $\cX\df\cX_1\times\ldots\times\cX_n$, i.e., pairs of joint actions $\alpha$ and strategy profiles $x$. We will denote the elements of the state space $\cZ$ by $z$. 

The set $\mathcal{A}$ is endowed with the discrete topology, $\cX$ with its usual Euclidean topology, and $\cZ$ with the corresponding product topology. We also let $\Bor(\cZ)$ denote the Borel $\sigma$-field of $\cZ$, and $\mathfrak{P}(\cZ)$ the set of \emph{probability measures} (p.m.) on $\Bor(\cZ)$ endowed with the Prohorov topology, i.e., the topology of weak convergence. The learning algorithm of Table~\ref{Tb:ReinforcementLearning} defines an $\cZ$-valued Markov chain. Let $P_{\lambda}:\cZ\times\Bor(\cZ)\to[0,1]$ denote its transition probability function (t.p.f.), parameterized by $\lambda>0$. We refer to the process with $\lambda>0$ as the \emph{perturbed process}. Let also $P:\cZ\times\Bor(\cZ)\to[0,1]$ denote the t.p.f. of the \emph{unperturbed process}, i.e., when $\lambda=0$. We also define the $t$-step t.p.f. $P^t:\cZ\times\Bor(\cZ)\to[0,1]$ recursively as: $$P^t(z,D)=\int_{\cZ}P(z,dy)P^{t-1}(y,D).$$

We let $\Cc(\cZ)$ denote the Banach space of real-valued continuous functions on $\cZ$ under the sup-norm (denoted by $\|\cdot\|_{\infty}$) topology. For $f\in\Cc(\cZ)$, define
\begin{equation*}
P_{\lambda}f(z) \df \int_{\cZ}P_{\lambda}(z,dy)f(y),
\end{equation*}
and 
\begin{equation*}
\mu[f] \df \int_{\cZ}\mu(dz)f(z), \mbox{ for } \mu\in\mathfrak{P}(\cZ).
\end{equation*}

The unperturbed process governed by the t.p.f. $P$ will be denoted by $Z\df\{Z_{t} : t\ge0\}$. Let $\Omega\df\cZ^{\infty}$ denote the canonical path space, i.e., an element $\omega\in\Omega$ is a sequence $\{\omega(0),\omega(1),\dotsc\}$, with $\omega(t)= (\alpha(t),x(t))\in\cZ$. We use the same notation for the elements $(\alpha,x)$ of the space $\cZ$ and for the coordinates of the process $Z_{t}=(\alpha(t),x(t))$.
Let also $\Prob_{z}[\cdot]$ denote the unique p.m. induced by the unperturbed process $P$ on the product $\sigma$-field of $\cZ^{\infty}$, initialized at $z=(\alpha,x)$, and $\Exp_{z}[\cdot]$ the corresponding expectation operator. Let also $\theta:\Omega\to\Omega$ denote the \emph{shift operator}, such that $(Z\circ\theta_{t})(\omega)\df Z(\theta_{t}(\omega))=\{Z_t,Z_{t+1},...\}$. 
Furthermore, for $D\in\Bor(\cZ)$, let $\uptau(D)$ be the first hitting time of the unperturbed process to $D$, i.e., $\uptau(D)\df\inf\{t\geq{0}:Z_t\in{D}\}$.  

\subsection{Stochastic stability}	

First, we note that both $P$ and $P_{\lambda}$ ($\lambda>0$) satisfy the \emph{weak Feller property} (cf.,~\cite[Definition~4.4.2]{Lerma03}).
\begin{proposition}		\label{Pr:WeakFeller}
Both the unperturbed process $P$ ($\lambda=0$) and the perturbed process $P_{\lambda}$ ($\lambda>0$) have the weak Feller property.
\end{proposition}
\begin{proof}
See Appendix~\ref{Ap:WeakFeller}.
\end{proof}

The measure $\mu_{\lambda}\in\mathfrak{P}(\cZ)$ is called an \emph{invariant probability measure} (i.p.m.) for $P_{\lambda}$ if
\begin{equation*}
(\mu_{\lambda}P_{\lambda})(A) \df \int_{\cZ}\mu_{\lambda}(dz)P_{\lambda}(z,A) = \mu_{\lambda}(A), \qquad A\in\Bor(\cZ).
\end{equation*}
Since $\cZ$ defines a locally compact separable metric space and $P$, $P_{\lambda}$ have the weak Feller property, they both admit an i.p.m., denoted $\mu$ and $\mu_{\lambda}$, respectively \cite[Theorem~7.2.3]{Lerma03}.

We would like to characterize the \emph{stochastically stable states} $z\in\cZ$ of $P_{\lambda}$, that is any state $z\in\cZ$ for which any collection of i.p.m.'s $\{\mu_{\lambda}\in\mathfrak{P}(\cZ):\mu_{\lambda}P_{\lambda}=\mu_{\lambda},\lambda>0\}$ satisfies $\liminf_{\lambda\to{0}}\mu_{\lambda}(z)>0$. As the forthcoming analysis will show, the stochastically stable states will be a subset of the set of \emph{pure strategy states} (p.s.s.) defined as follows:
\begin{definition}[Pure Strategy State]	\label{def:PureStrategyState}
\textit{
A pure strategy state is a state $s=(\alpha,x)\in\cZ$ such that for all $i\in\mathcal{I}$, $x_i = e_{\alpha_i}$, i.e., $x_i$ coincides with the vertex of the probability simplex $\Delta(\magn{\mathcal{A}_i})$ which assigns probability 1 to action $\alpha_i$.
}
\end{definition}

We will denote the set of pure strategy states by $\mathcal{S}$.

\begin{theorem}[Stochastic Stability]		\label{Th:StochasticStability}
There exists a unique probability vector $\pi=(\pi_1,...,\pi_{\magn{\mathcal{S}}})$ such that for any collection of i.p.m.'s $\{\mu_{\lambda}\in\mathfrak{P}(\cZ):\mu_{\lambda}P_{\lambda}=\mu_{\lambda}, \lambda>0\}$, the following hold:
\begin{itemize}
\item[(a)] $\lim_{\lambda\to{0}}\mu_{\lambda}(\cdot) = \hat{\mu}(\cdot) \df \sum_{s\in\mathcal{S}}\pi_s\Dirac{s}(\cdot),$ where convergence is in the weak sense.
\item[(b)] The probability vector $\pi$ is an invariant distribution of the (finite-state) Markov process $\hat{P}$, such that, for any $s,s'\in\mathcal{S}$,
\begin{equation}	\label{eq:FiniteStateMarkovChain}
\hat{P}_{ss'} \df \lim_{t\to\infty} QP^t(s,\Neig_{\delta}(s')),
\end{equation}
for some $\delta>0$ sufficiently small, where $Q$ is the t.p.f. corresponding to \emph{only one agent trembling} (i.e., following the uniform distribution of (\ref{eq:ActionUpdate})).
\end{itemize}
\end{theorem}

The proof of Theorem~\ref{Th:StochasticStability} requires a series of propositions and it will be presented in detail in Section~\ref{sec:TechnicalDerivation}. 

Theorem~\ref{Th:StochasticStability} implicitly provides a stochastically stability argument. In fact, the expected asymptotic behavior of the dynamics can be characterized by $\Hat{\mu}$ and, therefore, $\pi$. In particular, by Birkhoff's individual ergodic theorem \cite[Theorem~2.3.4]{Lerma03}, the weak convergence of $\mu_{\lambda}$ to $\Hat{\mu}$, and the fact that $\mu_{\lambda}$ is ergodic, we have that the expected percentage of time that the process spends in any $O\in\mathcal{B}(\cZ)$ such that
$\partial{O}\cap\mathcal{S}\ne\varnothing$ is given by $\Hat{\mu}(O)$ as the
experimentation probability $\lambda$ approaches zero and time
increases, i.e.,
\begin{equation*}
\lim_{\lambda\downarrow{0}}\left(\lim_{t\to\infty}\;
\frac{1}{t}\sum_{k=0}^{t-1}P_{\lambda}^{k}(x,O)\right) = \Hat{\mu}(O)\,.
\end{equation*}

\subsection{Discussion}

Theorem~\ref{Th:StochasticStability} establishes ``\textit{equivalence}'' (in a weak convergence sense) of the original (perturbed) learning process with a simplified process, where \emph{only one agent trembles at the first iteration and then no agent trembles thereafter}. This  simplification in the analysis has originally been capitalized to analyze \emph{aspiration learning} dynamics in \cite{Karandikar98,ChasparisAriShamma13_SIAM}, and it is based upon the observation that \emph{under the unperturbed process, agents' strategies will converge to a pure strategy state}, as it will be shown in the forthcoming Section~\ref{sec:TechnicalDerivation}.

The limiting behavior of the original (perturbed) dynamics is characterized by the (\emph{unique}) invariant distribution of the finite-state Markov chain $\{P_{ss'}\}$, whose states correspond to the pure strategy states $\cS$. In other words, \emph{we should expect that as the perturbation parameter $\lambda$ approaches zero, the algorithm spends the majority of the time on pure strategy states}. The importance of this result lies on the fact that no constraints have been imposed in the payoff matrix/function of the game other than the Positive-Utility Property \ref{P:PositiveUtilityProperty}. 

In the forthcoming Section~\ref{sec:StochasticallyStableStates}, we will use this result to provide a methodology for computing the set of stochastically stable states. Then, in Section~\ref{sec:Illustration}, this methodology will be illustrated in the context of coordination games.

\section{Technical Derivation}	\label{sec:TechnicalDerivation}

In this section, we provide the main steps for the proof of Theorem~\ref{Th:StochasticStability}. We begin by investigating the asymptotic behavior of the unperturbed process $P$, and then we characterize the i.p.m. of the perturbed process with respect to the p.s.s.'s $\cS$.

\subsection{Unperturbed Process}	\label{Sc:UnperturbedProcess}

For $t\ge0$ define the sets
\begin{align*}
A_{t} &\df \left\{\omega\in\Omega:\alpha(\tau)=\alpha(t)\,, \text{~for all~} \tau \geq t \right\}\,,\\[5pt]
B_{t} &\df \{\omega\in\Omega:\alpha(\tau)=\alpha(0)\,, \text{~for all~} 0\le\tau\le{t}\}\,.
\end{align*}
Note that $\{B_{t}:t\ge0\}$ is a non-increasing sequence, i.e., $B_{t+1}\subseteq B_{t}$, while $\{A_{t}:t\ge0\}$ is non-decreasing, i.e., $A_{t+1}\supseteq A_t$. Let
$A_{\infty} \df\bigcup_{t=0}^{\infty}A_{t} \mbox{ and } B_{\infty}\df\bigcap_{t=1}^{\infty}B_{t}.$
In other words, \emph{$A_{\infty}$ corresponds to the event that agents eventually play the same action profile, while $B_{\infty}$ corresponds to the event that agents never change their actions}. The following proposition discusses convergence of the unperturbed process to the set of p.s.s.'s, $\cS$.

\begin{proposition}[Convergence to p.s.s.]	\label{Pr:ConvergenceToPSS}
Let us assume that the step-size $\epsilon>0$ is sufficiently small such that $0<\epsilon u_i(\alpha)<1$ for all $\alpha\in\mathcal{A}$ and $i\in\mathcal{I}$. Then, the following hold: 
\begin{itemize}
\item[(a)] $\inf_{z\in\cZ}\;\mathbb{P}_{z}[B_{\infty}]>0$, 
\item[(b)] $\inf_{z\in\cZ}\Prob_{z}[A_\infty]=1$.
\end{itemize}
\end{proposition}
\begin{proof}
See Appendix~\ref{Ap:ConvergenceToPSS}.
\end{proof}

Statement (a) of Proposition~\ref{Pr:ConvergenceToPSS} states that \emph{the probability that agents never change their actions is bounded away from zero}, while statement (b) states that \emph{the probability that eventually agents play the same action profile is one}. This also indicates that any invariant measure of the unperturbed process can be characterized with respect to the pure strategy states $\cS$, which is established by the following proposition. 

\begin{proposition}[Limiting t.p.f. of unperturbed process]	\label{Pr:LimitingUnperturbedTPF}
Let $\mu$ denote an i.p.m. of $P$. Then, there exists a t.p.f. $\Pi$ on $\cZ\times\Bor(\cZ)$ with the following properties:
\begin{itemize}
\item[(a)] for $\mu$-a.e. $z\in\cZ$, $\Pi(z,\cdot)$ is an i.p.m. for $P$;
\item[(b)] for all $f\in\Cc(\cZ)$, $\lim_{t\to\infty}\|P^tf-\Pi f\|_{\infty}=0$;
\item[(c)] $\mu$ is an i.p.m. for $\Pi$;
\item[(d)] the support\footnote{The \emph{support} of a measure $\mu$ on $\cZ$ is the unique closed set $F\subset\Bor(\cZ)$ such that $\mu(\cZ\backslash{F})=0$ and $\mu(F\cap{O})>0$ for every open set $O\subset\cZ$ such that $F\cap{O}\neq\varnothing$.} of $\Pi$ is on $\mathcal{S}$ for all $z\in\cZ$.
\end{itemize}
\end{proposition}
\begin{proof}
The state space $\cZ$ is a locally compact separable metric space and the t.p.f. of the unperturbed process $P$ admits an i.p.m. due to Proposition~\ref{Pr:WeakFeller}. Thus, statements (a), (b) and (c) follow directly from \cite[Theorem~5.2.2 (a), (b), (e)]{Lerma03}. 

(d) Let us assume that the support of $\Pi$ includes points in $\cZ$ other than the pure strategy states in $\cS$. Then, there exists an open set $O\in\Bor(\cZ)$ such that $O\cap\cS=\varnothing$ and $\Pi(z^*,O)>0$ for some $z^*\in\cZ$. According to (b), $P^{t}$ converges weakly to $\Pi$. Thus, from Portmanteau theorem (cf.,~\cite[Theorem~1.4.16]{Lerma03}), we have that $\liminf_{t\to\infty} P^{t}(z^*,O) \geq \Pi(z^*,O)>0.$ This is a contradiction of Proposition~\ref{Pr:ConvergenceToPSS}(b), which concludes the proof.
\end{proof}

Proposition~\ref{Pr:LimitingUnperturbedTPF} states that the limiting unperturbed t.p.f. converges weakly to a t.p.f. $\Pi$ which accepts the same i.p.m. as $P$. Furthermore, \emph{the support of $\Pi$ is the set of pure strategy states $\mathcal{S}$}. This is a rather important observation, since the limiting perturbed process can also be ``related'' (in a weak-convergence sense) to the t.p.f. $\Pi$, as it will be shown in the following section.

\subsection{Perturbed process}		\label{sec:InvariantPMPerturbedProcess}

According to the definition of perturbed learning automata of Table~\ref{Tb:ReinforcementLearning}, when an agent updates its action, there is a small probability $\lambda>0$ that it ``\emph{trembles},'' i.e., it selects a new action according to a uniform distribution. Thus, we can decompose the t.p.f. induced by the one-step update as follows:
$$P_{\lambda} = (1-\varphi(\lambda))P + \varphi(\lambda)Q_{\lambda}$$ where $\varphi(\lambda) = 1-(1-\lambda)^{n}$ is the probability that at least one agent trembles (since $(1-\lambda)^{n}$ is the probability that no agent trembles), and $Q_{\lambda}$ corresponds to the t.p.f. when at least one agent trembles. Note that $\varphi(\lambda)\to{0}$ as $\lambda\downarrow{0}$. 

Define also $Q$ as the t.p.f. where \emph{only one} agent trembles, and $Q^*$ as the t.p.f. where \emph{at least two agents tremble}. Then, we may write:
\begin{equation}
Q_{\lambda} = (1-\psi(\lambda))Q + \psi(\lambda)Q^*,
\end{equation}
where $\psi(\lambda) \df 1-\frac{n\lambda(1-\lambda)^{n-1}}{1-(1-\lambda)^{n}}$ corresponds to the probability that at least two agents tremble given that at least one agent trembles. It also satisfies $\psi(\lambda)\to{0}$ as $\lambda\downarrow{0}$, which establishes an approximation of $Q_{\lambda}$ by $Q$ as the perturbation factor $\lambda$ approaches zero.

Let us also define the infinite-step t.p.f. when trembling only at the first step (briefly, \emph{lifted} t.p.f.) as follows: 
\begin{equation}
P_{\lambda}^{L} \df \varphi(\lambda)\sum_{t=0}^{\infty}(1-\varphi(\lambda))^{t}Q_{\lambda}P^{t} = Q_{\lambda}R_{\lambda}
\end{equation}
where
$R_{\lambda} \df \varphi(\lambda)\sum_{t=0}^{\infty}(1-\varphi(\lambda))^{t}P^{t},$
i.e., $R_{\lambda}$ corresponds to the \emph{resolvent} t.p.f. 

In the following proposition, we establish weak-convergence of the lifted t.p.f. $P_{\lambda}^{L}$ to $Q\Pi$ as $\lambda\downarrow{0}$, which will further allow for an explicit characterization of the weak limit points of the i.p.m. of $P_{\lambda}$.

\begin{proposition}[i.p.m. of perturbed process]		\label{Pr:WeakLimitPointsOfPerturbedInvariantMeasures}
The following hold:
\begin{itemize}
\item[(a)] For $f\in\Cc(\cZ)$, $\lim_{\lambda\to{0}}\|R_{\lambda}f-\Pi{f}\|_{\infty} = 0.$
\item[(b)] For $f\in\Cc(\cZ)$, $\lim_{\lambda\to{0}}\|P_{\lambda}^{L}f-Q\Pi{f}\|_{\infty} = 0$.
\item[(c)] Any i.p.m. $\mu_{\lambda}$ of $P_\lambda$ is also an i.p.m. of $P_{\lambda}^{L}$.
\item[(d)] Any weak limit point in $\mathfrak{P}(\cZ)$ of $\mu_{\lambda}$, as $\lambda\to{0}$, is an i.p.m. of $Q\Pi$.
\end{itemize}
\end{proposition}
\begin{proof}
(a) For any $f\in \Cc(\cZ)$, we have
\begin{eqnarray*}
\lefteqn{\|R_{\lambda}f - \Pi{f}\|_{\infty}} \cr 
& = & \supnorm{\varphi(\lambda)\sum_{t=0}^{\infty}(1-\varphi(\lambda))^tP^tf - \Pi f}  \cr 
& = & \supnorm{\varphi(\lambda)\sum_{t=0}^{\infty}(1-\varphi(\lambda))^t(P^t f - \Pi f)},
\end{eqnarray*}
where we have used $\varphi(\lambda)\sum_{t=0}^{\infty}(1-\varphi(\lambda))^t=1$. Note that
\begin{eqnarray*}
\lefteqn{\varphi(\lambda)\sum_{t=T}^{\infty}(1-\varphi(\lambda))^t\supnorm{P^tf-\Pi f} } \cr
& \leq & (1-\varphi(\lambda))^{T}\sup_{t\geq{T}}\supnorm{P^tf - \Pi f}. 
\end{eqnarray*}
From Proposition~\ref{Pr:LimitingUnperturbedTPF}(b), we have that for any $\delta>0$, there exists $T=T(\delta)>0$ such that the r.h.s. is uniformly bounded by $\delta$ for all $t\geq T$. Thus, the sequence $$\Theta_T\df\varphi(\lambda)\sum_{t=0}^{T}(1-\varphi(\lambda))^t(P^tf - \Pi f)$$ is Cauchy and therefore convergent (under the sup-norm). In other words, there exists $\Theta\in\mathbb{R}$ such that $\lim_{T\to\infty}\supnorm{\Theta_{T}-\Theta}=0.$ 
For every $T>0$, we have
\begin{eqnarray*}
\supnorm{R_{\lambda}f-\Pi{f}} \leq \supnorm{\Theta_{T}} + \supnorm{\Theta - \Theta_T}.
\end{eqnarray*}
Note that 
\begin{eqnarray*}
\supnorm{\Theta_T} \leq \varphi(\lambda) \sum_{t=0}^{T} (1-\varphi(\lambda))^{t}  \supnorm{P^t f-\Pi f}.
\end{eqnarray*}
If we take $\lambda\downarrow{0}$, then the r.h.s. converges to zero. Thus,
\begin{equation*}
\lim_{\lambda\downarrow{0}}\|R_{\lambda}f-\Pi{f}\|_{\infty} \leq \supnorm{\Theta-\Theta_T}, \mbox{ for all } T>0,
\end{equation*}
which concludes the proof.

(b) For any $f\in \Cc(\cZ)$, we have
\begin{eqnarray*}
\lefteqn{ \|P_{\lambda}^{L}f-Q\Pi{f}\|_{\infty} } \cr
& \leq & \|Q_{\lambda}(R_{\lambda}f-\Pi{f})\|_{\infty} + \|Q_{\lambda}\Pi{f} - Q\Pi{f}\|_{\infty} \cr
& \leq & \|R_{\lambda}f-\Pi{f}\|_{\infty} + \|Q_{\lambda}\Pi{f}-Q\Pi{f}\|_{\infty}.
\end{eqnarray*}
The first term of the r.h.s. approaches 0 as $\lambda\downarrow{0}$ according to (a). The second term of the r.h.s. also approaches 0 as $\lambda\downarrow{0}$ since $Q_{\lambda}\rightarrow{Q}$ as $\lambda\downarrow{0}$.

(c) By definition of the perturbed t.p.f. $P_{\lambda}$, we have
\begin{equation*}
P_{\lambda}R_{\lambda} = (1-\varphi(\lambda))PR_{\lambda} + \varphi(\lambda)Q_{\lambda}R_{\lambda}.
\end{equation*}
Note that $Q_{\lambda}R_{\lambda}=P_{\lambda}^{L}$ and $(1-\varphi(\lambda))PR_{\lambda} = R_{\lambda} - \varphi(\lambda)I,$ where $I$ corresponds to the identity operator. Thus, 
\begin{equation*}
P_{\lambda}R_{\lambda} = R_{\lambda}-\varphi(\lambda)I+\varphi(\lambda)P_{\lambda}^{L}.
\end{equation*}
For any i.p.m. of $P_{\lambda}$, $\mu_{\lambda}$, we have
\begin{equation*}
\mu_{\lambda}P_{\lambda}R_{\lambda} = \mu_{\lambda}R_{\lambda}-\varphi(\lambda)\mu_{\lambda}+\varphi(\lambda)\mu_{\lambda}P_{\lambda}^{L},
\end{equation*}
which equivalently implies that $\mu_{\lambda} = \mu_{\lambda}P_{\lambda}^{L},$ since $\mu_{\lambda}P_{\lambda} = \mu_{\lambda}$. We conclude that $\mu_{\lambda}$ is also an i.p.m. of $P_{\lambda}^{L}$.

(d) Let $\hat{\mu}$ denote a weak limit point of $\mu_{\lambda}$ as $\lambda\downarrow{0}$. To see that such a limit exists, take $\hat{\mu}$ to be an i.p.m. of $P$. Then, 
\begin{eqnarray*}
\lefteqn{\|P_{\lambda}f-P{f}\|_{\infty}} \cr & \geq & \|\mu_{\lambda}(P_{\lambda}f-P{f})\|_{\infty} \cr & = & \|(\mu_{\lambda}-\hat{\mu})(I-P)[f]\|_{\infty}.
\end{eqnarray*} 
Thus, the weak convergence of $P_{\lambda}$ to $P$ implies that $\mu_{\lambda}\Rightarrow\hat{\mu}$. Note further that
\begin{eqnarray*}
\lefteqn{\hat{\mu}[f] - \hat{\mu}[Q\Pi{f}]}\cr & = & (\hat{\mu}[f]-\mu_{\lambda}[f]) + \mu_{\lambda}[P_{\lambda}^{L}f-Q\Pi{f}]+ \cr && (\mu_{\lambda}[Q\Pi{f}]-\hat{\mu}[Q\Pi{f}]).
\end{eqnarray*}
The first and the third term of the r.h.s. approaches 0 as $\lambda\downarrow{0}$ due to the fact that $\mu_{\lambda}\Rightarrow\hat{\mu}$. The same holds for the second term of the r.h.s. due to part (b). Thus, we conclude that any weak limit point of $\mu_{\lambda}$ as $\lambda\downarrow{0}$ is an i.p.m. of $Q\Pi$.
\end{proof}

Proposition~\ref{Pr:WeakLimitPointsOfPerturbedInvariantMeasures} establishes convergence (in a weak sense) of the i.p.m. $\mu_{\lambda}$ of the perturbed process to an i.p.m. of $Q\Pi$. In the following section, this convergence result will allow for a more explicit characterization of $\mu_{\lambda}$ as $\lambda\downarrow{0}$.

\subsection{Equivalent finite-state Markov process}	\label{sec:FiniteStateMarkovChainEquivalence}

Define the finite-state Markov process $\hat{P}$ as in (\ref{eq:FiniteStateMarkovChain}). 

\begin{proposition} [Unique i.p.m. of $Q\Pi$]		\label{Pr:UniqueInvariantPMofQPi}
There exists a unique i.p.m. $\hat{\mu}$ of $Q\Pi$. It satisfies 
\begin{equation}	\label{eq:InvariantMeasureQP_derivation}
\hat{\mu}(\cdot) = \sum_{s\in\mathcal{S}}\pi_s\Dirac{s}(\cdot)
\end{equation}
for some constants $\pi_s\geq{0}$, $s\in\mathcal{S}$. Moreover, $\pi=(\pi_1,...,\pi_{\magn{\mathcal{S}}})$ is an invariant distribution of $\hat{P}$, i.e., $\pi=\pi\hat{P}$.
\end{proposition}
\begin{proof}
From Proposition~\ref{Pr:LimitingUnperturbedTPF}(d), we know that the support of $\Pi$ is the set of pure strategy states $\mathcal{S}$. Thus, the support of $Q\Pi$ is also on $\mathcal{S}$. From Proposition~\ref{Pr:WeakLimitPointsOfPerturbedInvariantMeasures}, we know that $Q\Pi$ admits an i.p.m., say $\hat{\mu}$, whose support is also $\mathcal{S}$. Thus, $\hat{\mu}$ admits the form of (\ref{eq:InvariantMeasureQP_derivation}), for some constants $\pi_{s}\geq{0}$, $s\in\mathcal{S}$.

For any two distinct $s,s'\in\cS$, note that $\Neig_{\delta}(s')$, $\delta>0$, is a continuity set of $Q\Pi(s,\cdot)$, i.e., $Q\Pi(s,\partial\Neig_{\delta}(s'))=0$. Thus, from Portmanteau theorem, given that $QP^{t}\Rightarrow Q\Pi$, $$Q\Pi(s,\Neig_{\delta}(s')) = \lim_{t\to\infty}QP^{t}(s,\Neig_{\delta}(s')) = \hat{P}_{ss'}.$$  If we also define $\pi_s \df \hat{\mu}(\Neig_{\delta}(s))$, then
$$\pi_{s'} = \hat{\mu}(\Neig_{\delta}(s')) = \sum_{s\in\mathcal{S}}\pi_s Q\Pi(s,\Neig_{\delta}(s')) = \sum_{s\in\mathcal{S}}\pi_s\hat{P}_{ss'},$$ which shows that $\pi$ is an invariant distribution of $\hat{P}$.

It remains to establish uniqueness of the invariant distribution of $Q\Pi$. Note that the set $\mathcal{S}$ of pure strategy states is isomorphic with the set $\mathcal{A}$ of action profiles. If agent $i$ trembles (as t.p.f. $Q$ dictates), then all actions in $\mathcal{A}_i$ have positive probability of being selected, i.e., $Q(\alpha,(\alpha_i',\alpha_{-i}))>0$ for all $\alpha_i'\in\mathcal{A}_i$ and $i\in\cI$. It follows by Proposition~\ref{Pr:ConvergenceToPSS} that $Q\Pi(\alpha,(\alpha_i',\alpha_{-i}))>0$ for all $\alpha_i'\in\mathcal{A}_i$ and $i\in\cI$. Finite induction then shows that $(Q\Pi)^{n}(\alpha,\alpha')>0$ for all $\alpha,\alpha'\in\cA$. It follows that if we restrict the domain of $Q\Pi$ to $\mathcal{S}$, it defines an irreducible stochastic matrix. Therefore, $Q\Pi$ has a unique i.p.m.
\end{proof}

\subsection{Proof of Theorem~\ref{Th:StochasticStability}}	\label{sec:Proof:StochasticStability}

Theorem~\ref{Th:StochasticStability}(a)--(b) is a direct implication of Propositions~\ref{Pr:WeakLimitPointsOfPerturbedInvariantMeasures}--\ref{Pr:UniqueInvariantPMofQPi}.

\section{Stochastically Stable States}	\label{sec:StochasticallyStableStates}

In this section, we capitalize on Theorem~\ref{Th:StochasticStability} and we further simplify the computation of the stochastically stable states in strategic-form games satisfying Property~\ref{P:PositiveUtilityProperty}. 

\subsection{Background on finite Markov chains}	\label{sec:FiniteMarkovChains}

In order to compute the invariant distribution of a finite-state, irreducible and aperiodic Markov chain, we are going to consider a characterization introduced by \cite{FreidlinWentzell84}. In particular, for finite Markov chains an invariant measure can be expressed as the ratio of sums of products consisting of transition probabilities. These products can be described conveniently by means of graphs on the set of states of the chain. In particular, let  $\mathcal{S}$ be a finite set of states, whose elements will be denoted by $s_k$, $s_\ell$, etc., and let a subset $\mathcal{W}$ of $\mathcal{S}$.

\begin{figure}[t!]
\centering
\includegraphics[scale=1]{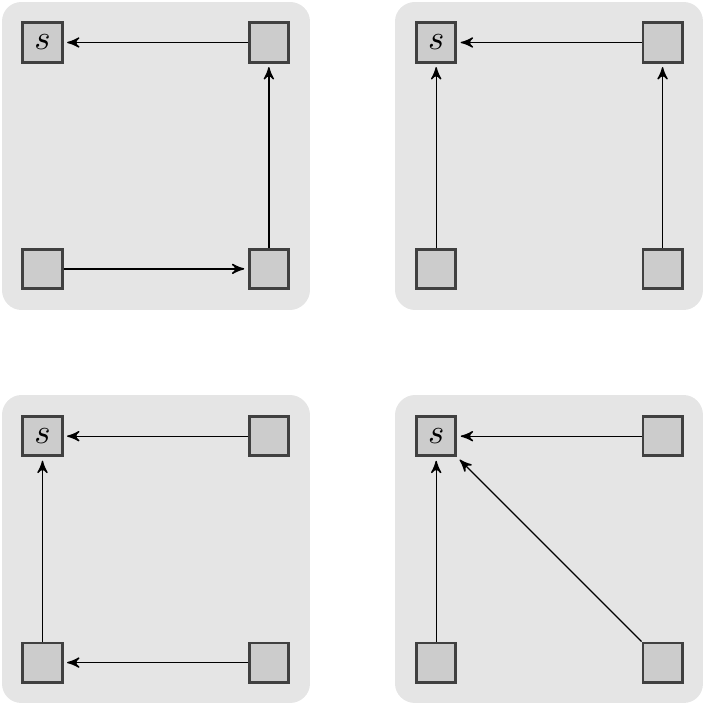}
\caption{Examples of $s$-graphs in case $\mathcal{S}$ contains four states.}
\label{fig:Wgraphs}
\end{figure}

\begin{definition}  \label{Df:W-graph}
($\mathcal{W}$-graph) A graph consisting of arrows
$s_k\rightarrow{s_\ell}$
($s_k\in{\mathcal{S}\backslash{\mathcal{W}}},s_\ell\in\mathcal{S},s_\ell\neq{s_k}$)
is called a $\mathcal{W}$-graph if it satisfies the following
conditions:
\begin{enumerate}
\item{every point $k\in{\mathcal{S}\backslash{\mathcal{W}}}$ is the initial point of
exactly one arrow;}
\item{there are no closed cycles in the graph; or, equivalently,
for any point $s_k\in{\mathcal{S}\backslash{\mathcal{W}}}$ there
exists a sequence of arrows leading from it to some point
$s_\ell\in\mathcal{W}$.}
\end{enumerate}
\end{definition}

\figurename~\ref{fig:Wgraphs} provides examples of $\{s\}$-graphs for some state $s\in\cS$ when $\cS$ contains four states. We will denote by $\mathcal{G}\{\mathcal{W}\}$ the set of $\mathcal{W}$-graphs and we shall use the letter $g$ to denote graphs. If $\hat{P}_{s_ks_\ell}$ are nonnegative numbers, where $s_k,s_\ell\in\mathcal{S}$, define also the transition probability along path $g$ as
\begin{equation*}
  \varpi(g) \df
  \prod_{(s_k\rightarrow{s_\ell})\in{g}}\hat{P}_{s_ks_\ell}.
\end{equation*}

The following Lemma holds:
\begin{lemma}[Lemma~6.3.1 in \cite{FreidlinWentzell84}]   \label{Lm:StationaryDistribution}
Let us consider a Markov chain with a finite set of states $\mathcal{S}$ and transition probabilities $\{\hat{P}_{s_ks_\ell}\}$ and assume that every state can be reached from any other state in a finite number of steps. Then, the stationary distribution of the chain is $\pi = [\pi_{s}]$, where
\begin{equation}    \label{eq:StationaryDistribution}
\pi_s = \frac{R_{s}}{\sum_{s_i\in\mathcal{S}}R_{s_i}},
s\in\mathcal{S}
\end{equation}
where $R_{s} \df \sum_{g\in{\mathcal{G}}\{s\}}\varpi(g)$.
\end{lemma}

In other words, in order to compute the weight that the stationary distribution assigns to a state $s\in\cS$, it suffices to compute the ratio of the transition probabilities of all $\{s\}$-graphs over the transition probabilities of all graphs.

\subsection{Approximation of one-step transition probability}

We wish to provide an approximation in the computation of the transition probabilities under $\hat{P}$, in order to explicitly compute the stationary distribution $\pi$ of Theorem~\ref{Th:StochasticStability}. Based on the definition of the t.p.f. $Q\Pi$, and as $\lambda\downarrow{0}$, a transition from $s$ to $s'$ influences the stationary distribution only if $s$ differs from $s'$ in the action of a \emph{single} agent. This observation will be capitalized by the forthcoming Lemmas~\ref{Lm:TransitionProbabilities}--\ref{Lm:StationaryDistributionApproximation}, to approximate the transition probability from $s$ to $s'$ under $\hat{P}$.


Let us define $\uptau_s^*(D)$ to be the minimum number of steps that the process $Q\Pi$ needs in order to reach $D$ when starting from $s\in\cS$ (i.e., the minimum possible first hitting time to $D$).

\begin{lemma}[One-step transition probability]	\label{Lm:TransitionProbabilities}
Consider any two action profiles $\alpha,\alpha'\in\mathcal{A}$ which differ in the action of a single agent $j$, and let $s,s'\in\cS$ be the p.s.s.'s associated with $\alpha$ and $\alpha'$, respectively. Set $z'=(\alpha',x')$, where $x_j'\df e_{\alpha_j}+\epsilon u_{j}(\alpha')(e_{\alpha_j'}-e_{\alpha_j}),$ which corresponds to the state after agent $j$ perturbed once starting from $s$ and played $\alpha_j'$. Define also $$\breve{P}_{ss'}(\delta) \df \Prob_{z'}[\uptau(\Neig_{\delta}(s')) \leq \infty]$$ which corresponds to the probability that the process eventually reaches $\Neig_{\delta}(s')$ starting from $z'$. For sufficiently small $\epsilon$ such that $0<\epsilon u_j(\alpha') <1$, the following hold:
\begin{itemize}

\item[(a)] The transition probability from $s$ to $s'$ under $Q\Pi$ can be approximated as follows:
\begin{equation}	\label{eq:TransitionProbabilityApproximation_Initial}
\hat{P}_{ss'} = \gamma_{j} \cdot \lim_{\delta\downarrow{0}} \breve{P}_{ss'}(\delta),
\end{equation}
where $\gamma_{j}\df 1/(n\magn{\cA_j})$ corresponds to the probability that agent $j$ trembled and selected action $\alpha_j'$ under $Q$.
\item[(b)] $\breve{P}_{ss'}(\delta)$ coincides with the probability of the \emph{shortest path}, i.e.,
\begin{equation*}
\breve{P}_{ss'}(\delta) = \Prob_{z'}\left[\alpha(t+1)=\alpha'\,, \forall t < \uptau_{s}^*(\Neig_{\delta}(s'))\right].
\end{equation*}
\item[(c)] There exists negative constant $\eta(\delta)$, such that for any transition step $s\to{s'}$ (with the above properties) and for sufficiently small $\epsilon>0$, 
\begin{equation} \label{eq:TransitionProbabilityApproximation}
\breve{P}_{ss'}(\delta)  \approx \exp\left(\frac{\eta(\delta)}{\epsilon u_j(\alpha')}\right).
\end{equation}
\end{itemize}
\end{lemma}
\begin{proof}
See Appendix~\ref{Ap:StationaryDistributionApproximation}.
\end{proof}

Note that for sufficiently small $\epsilon$, \emph{the larger the destination utility $u_j(\alpha')$, the larger the transition probability to $s'$}, since $\eta(\delta)<0$. In a way, the inverse of the destination utility at $s'$ represents a measure of ``resistance'' of the process to transit to $s'$. Lemma~\ref{Lm:TransitionProbabilities} provides a tool for simplifying the computation of stochastically stable pure strategy states as it will become apparent in the following section. 

\subsection{Approximation of stationary distribution}

In this section, using Lemma~\ref{Lm:TransitionProbabilities} that approximates one-step transition probabilities, we provide an approximation of the i.p.m. of the t.p.f. $Q\Pi$. By definition of $Q\Pi$, this approximation is based upon the observation that for the computation of the quantities $R_{s}$ of Lemma~\ref{Lm:StationaryDistribution}, it suffices to consider only those paths in $\mathcal{G}\{s\}$ which involve \textit{one-step} transitions as defined in Lemma~\ref{Lm:TransitionProbabilities}. 

Define $\mathcal{G}^{(1)}\{s\}\subseteq\mathcal{G}\{s\}$ to be the set of $s$-graphs consisting solely of one-step transitions, i.e., for any $g\in\mathcal{G}^{(1)}\{s\}$ and any arrow $(s_k\to s_{\ell})\in{g}$, the associated action profiles, say $\alpha^{(k)},\alpha^{(\ell)}$, respectively, differ in a single action of a single agent. It is straightforward to check that $\mathcal{G}^{(1)}\{s\}\neq\varnothing$ for any $s\in\mathcal{S}$. 

\begin{lemma}[Approximation of stationary distribution]	\label{Lm:StationaryDistributionApproximation}
The stationary distribution of the finite Markov chain $\{\hat{P}_{s_ks_{\ell}}\}$, $\pi=[\pi_s]$, satisfies
\begin{equation}    \label{eq:StationaryDistributionSimplified}
\pi_s = \lim_{\delta\downarrow{0}}\frac{\breve{R}_{s}(\delta)}{\sum_{s_i\in\mathcal{S}}\breve{R}_{s_i}(\delta)}, \qquad
s\in\mathcal{S},
\end{equation}
where 
$\breve{R}_{s}(\delta) \df \sum_{g\in{\mathcal{G}^{(1)}}\{s\}}\breve{\varpi}(g;\delta),$ and
\begin{equation}	\label{eq:TransitionProbabilityGraphApproximation}
\breve{\varpi}(g;\delta) \df \bar{\gamma}_g \prod_{(s_k\to s_{\ell})\in{g}}\breve{P}_{s_ks_{\ell}}(\delta),
\end{equation}
for some constant $\bar{\gamma}_g \in (0,1)$.
\end{lemma}
\begin{proof}
According to Lemma~\ref{Lm:StationaryDistribution}, for any $s\in\cS$, we have $\pi_s=R_{s}/\sum_{s_i\in\cS}R_{s_i}$. Given the definition of the t.p.f. $Q$, where only one agent trembles, we should only consider one-step transition probabilities (as defined in Lemma~\ref{Lm:TransitionProbabilities}). Thus,
$$R_s = \sum_{g\in\mathcal{G}^{(1)}\{s\}}\varpi(g) = \sum_{g\in\mathcal{G}^{(1)}\{s\}}\prod_{(s_k\to s_{\ell})\in{g}}\hat{P}_{s_ks_{\ell}}.$$
According to Lemma~\ref{Lm:TransitionProbabilities} and Equation~(\ref{eq:TransitionProbabilityApproximation_Initial}), we have 
\begin{eqnarray*}
R_s & = & \lim_{\delta\downarrow{0}}\sum_{g\in\mathcal{G}^{(1)}\{s\}}\prod_{(s_k\to s_{\ell})\in{g}} \gamma_{j(s_k,s_{\ell})}\breve{P}_{s_k s_{\ell}}(\delta) \cr 
& = & \lim_{\delta\downarrow{0}}\sum_{g\in\mathcal{G}^{(1)}\{s\}}\bar{\gamma}_{g} \prod_{(s_k\to s_{\ell})\in{g}} \breve{P}_{s_k s_{\ell}}(\delta)
\end{eqnarray*}
where $j(s_k,s_{\ell})$ denotes the single agent whose action changes from $s_k$ to $s_{\ell}$, and $\bar{\gamma}_g \df \prod_{(s_k\to s_{\ell})\in{g}} \gamma_{j(s_k,s_{\ell})}\in(0,1)$. Thus, the conclusion follows.
\end{proof}

Note that Lemma~\ref{Lm:StationaryDistributionApproximation} provides a simplification to Theorem~\ref{Th:StochasticStability}, since it suffices to compute the transition probabilities of the $\mathcal{W}$-graphs consisting solely of one step transitions. Furthermore, the transition probability of any such graph, $\breve{\varpi}(g;\delta)$, can be computed by Lemma~\ref{Lm:TransitionProbabilities}, which provides an explicit formula for one-step transitions. In the following section, the computation of the stationary distribution will further be simplified and related to the resistance of one-step transitions.

\subsection{$\delta$-resistance}

We have shown in Lemma~\ref{Lm:TransitionProbabilities}, that the one-step transition probability $\breve{P}_{ss'}(\delta)$ increases as the destination utility increases. Informally, \emph{the inverse destination utility at $s'$ represents a form of ``resistance'' to approaching state $s'$}. In this section, we will formalize this notion.
\begin{definition}[$\delta$-resistance]
For a pure strategy state $s\in\mathcal{S}$, let us consider any graph $g\in\mathcal{G}^{(1)}\{s\}$. For any $\delta>0$, the $\delta$-resistance associated with $s\in\mathcal{S}$ in graph $g$, is defined as follows:
\begin{equation}	\label{eq:DeltaResistance}
\varphi_{\delta}(s|g) \df \sum_{(s_k\rightarrow{s_{\ell}})\in{g}}\frac{1}{\epsilon u_j(\alpha^{(\ell)})}.
\end{equation}
\end{definition}

In other words, the $\delta$-resistance of a state $s$ along a graph $g$ corresponds to the sum of the inverse utilities of the destination states along that graph, scaled by $1/\epsilon$. We further denote by $\varphi_\delta^{*}(s)$ the minimum $\delta$-resistance, i.e., $\varphi_{\delta}^*(s)\df\min_{g\in\mathcal{G}^{(1)}\{s\}}\varphi_{\delta}(s|g)$ and by $g^*(s)$ the $\{s\}$-graph that attains this minimum resistance. 

\subsection{Stochastically stable states}

The stochastically stable states can be identified as the states of minimum resistance.

\begin{theorem}[Stochastically stable states]	\label{Th:StochasticallyStableStatesMinimumResistance}
As $\epsilon\downarrow{0}$, the set of stochastically stable p.s.s.'s $\mathcal{S}^*$ is such that, for any $\delta>0$
\begin{equation}	\label{eq:StochasticStabilityProperty}
\overline{\Phi}_\delta(\cS^*)\df \max_{s^*\in\cS^*}\varphi^*_{\delta}(s^*) < \min_{s\in\cS\backslash\cS^*}\varphi^*_{\delta}(s) \df \underline{\Phi}_\delta(\cS\backslash\cS^*).
\end{equation}
\end{theorem}
\begin{proof}
By Lemmas~\ref{Lm:TransitionProbabilities}--\ref{Lm:StationaryDistributionApproximation}, for any state $s\in\cS$ and for any graph $g\in\mathcal{G}^{(1)}\{s\}$, we have that, as $\epsilon\downarrow{0}$, 
\begin{equation*}
\breve{\varpi}(g;\delta) = \bar{\gamma}_g \prod_{(s_k\to s_{\ell})\in{g}} \breve{P}_{s_k s_{\ell}}(\delta) \approx \bar{\gamma}_g e^{\eta(\delta)  \varphi_{\delta}(s|g)},
\end{equation*}
and
\begin{equation*}
\breve{R}_s(\delta) = \sum_{g\in\mathcal{G}^{(1)}\{s\}}\bar{\gamma}_g e^{\eta(\delta)  \varphi_{\delta}(s|g) }.
\end{equation*}
Thus, for the states in $\cS\backslash\cS^*$, and for sufficiently small $\epsilon\downarrow{0}$, we have
\begin{eqnarray*}
\lefteqn{\sum_{s\in\cS\backslash\cS^*}\breve{R}_s(\delta) = e^{\eta(\delta)  \underline{\Phi}_{\delta}(\cS\backslash\cS^*)} \cdot } \cr && \sum_{s\in\cS\backslash\cS^*}\sum_{g\in\mathcal{G}^{(1)}\{s\}} \bar{\gamma}_{g} e^{\eta(\delta)(\varphi_{\delta}(s|g)-\underline{\Phi}_{\delta}(\cS\backslash\cS^*))}.
\end{eqnarray*}
Note that the second part of the r.h.s. approaches a finite value as $\epsilon\downarrow{0}$, since $\varphi_{\delta}(s|g)\geq \underline{\Phi}_{\delta}(\cS\backslash\cS^*)$ for each $s\in\cS\backslash\cS^*$. Analogously, for the states in $\cS^*$, we have
\begin{eqnarray*}
\lefteqn{\sum_{s\in\cS^*}\breve{R}_s(\delta) = e^{\eta(\delta) \overline{\Phi}_\delta(\cS^*)} \cdot } \cr && \sum_{s\in\cS^*}\sum_{g\in\mathcal{G}^{(1)}\{s\}}\bar{\gamma}_g e^{\eta(\delta)(\varphi_{\delta}(s|g)-\overline{\Phi}_\delta(\cS^*))}.
\end{eqnarray*}
Thus, for sufficiently small $\epsilon$,
\begin{eqnarray*}
\lefteqn{\frac{\sum_{s\in\cS\backslash\cS^*}\breve{R}_s(\delta)}{\sum_{s\in\cS^*}\breve{R}_s(\delta)} = e^{\eta(\delta) ( \underline{\Phi}_{\delta}(\cS\backslash\cS^*) - \overline{\Phi}_\delta(\cS^*))}} \cdot \cr && \frac{\sum_{s\in\cS\backslash\cS^*}\sum_{g\in\mathcal{G}^{(1)}\{s\}} \bar{\gamma}_g e^{\eta(\delta)(\varphi_{\delta}(s|g)-\underline{\Phi}_{\delta}(\cS\backslash\cS^*))}}{\sum_{s\in\cS^*}\sum_{g\in\mathcal{G}^{(1)}\{s\}} \bar{\gamma}_g e^{\eta(\delta)(\varphi_{\delta}(s|g)-\overline{\Phi}_\delta(\cS^*))}}.
\end{eqnarray*}
Given that $\underline{\Phi}_{\delta}(\cS\backslash\cS^*) - \overline{\Phi}_\delta(\cS^*)>0$, the first part of the r.h.s. approaches $0$ as $\epsilon\downarrow{0}$. Also, the numerator of the ratio of the r.h.s. approaches a finite value, due to the definition of $\underline{\Phi}_{\delta}(\cS^*)$. On the other hand, each term of the denominator approaches either a finite value or $\infty$ as $\epsilon\downarrow{0}$. Thus, 
\begin{equation} 	\label{eq:Theorem5.1-1}
\frac{\sum_{s\in\cS\backslash\cS^*}\breve{R}_s(\delta)}{\sum_{s\in\cS^*}\breve{R}_s(\delta)}\xrightarrow{\epsilon\downarrow{0}}0.
\end{equation}
Denote by $\pi_{\cS^*}$ the probability assigned by the stationary distribution $\pi$ to $\cS^*$. Then, according to (\ref{eq:StationaryDistributionSimplified}), we have:
\begin{eqnarray*}
\lefteqn{\lim_{\epsilon\downarrow{0}}\pi_{\cS^*} } \cr & = & \lim_{\epsilon\downarrow{0}}\lim_{\delta\downarrow{0}}\frac{\sum_{s^*\in\cS^*}\breve{R}_{s^*}(\delta)}{\sum_{s\in\mathcal{S}}\breve{R}_{s}(\delta)} = \lim_{\delta\downarrow{0}}\lim_{\epsilon\downarrow{0}}\frac{\sum_{s^*\in\cS^*}\breve{R}_{s^*}(\delta)}{\sum_{s\in\mathcal{S}}\breve{R}_{s}(\delta)} \cr
& = & \lim_{\delta\downarrow{0}}\lim_{\epsilon\downarrow{0}}\frac{1}{1+\sum_{s\in\mathcal{S}\backslash\cS^*}\breve{R}_{s}(\delta)/\sum_{s^*\in\cS^*}\breve{R}_{s^*}(\delta)}.
\end{eqnarray*}
Note that the interchange of limits in the second equality is valid due to the finiteness of the limits of the transition probabilities (according to Lemma~\ref{Lm:StationaryDistributionApproximation}). Given (\ref{eq:Theorem5.1-1}), we conclude that $\lim_{\epsilon\downarrow{0}}\pi_{\cS^*}=1$. Conversely, $\lim_{\epsilon\downarrow{0}}\pi_{\cS\backslash\cS^*}=0$.
Thus, the stochastically stable states may only be contained in $\mathcal{S}^*$.
\end{proof}

In other words, Theorem~\ref{Th:StochasticallyStableStatesMinimumResistance} says that, in order for a p.s.s. set $\cS^*$ to be stochastically stable, it suffices to show that for any $s\in\cS^*$ there exists a $\{s\}$-graph with strictly smaller $\delta$-resistance from any other state $s'\in\cS\backslash\cS^*$. Note that this theorem applies to any game that satisfies the positive-utility property. In the following section, we illustrate the utility of Theorem~\ref{Th:StochasticallyStableStatesMinimumResistance} in computing the stochastically stable states in coordination games.

\section{Illustration in Coordination Games}		\label{sec:Illustration}

\subsection{Stochastic stability}

In this section, we will be using the notion of \emph{best response} of a agent $i$ into an action profile $\alpha=(\alpha_i,\alpha_{-i})$, as well as the notion of \emph{Nash equilibrium}. In particular, we define:

\begin{definition}[Best response]		\label{def:BestResponse}
The best response of a player $i$ to an action profile $\alpha=(\alpha_i,\alpha_{-i})$ is defined as the following set of actions: ${\rm BR}_i(\alpha)\df \arg\max_{a\in\mathcal{A}_i}u_i(a,\alpha_{-i}).$
\end{definition}

\begin{definition}[Nash equilibrium]		\label{def:NashEquilibrium}
An action profile $\alpha^*=(\alpha_i^*,\alpha_{-i}^*)$ is a Nash equilibrium, if for every player $i$, $\alpha_i^* \in {\rm BR}_i(\alpha^*).$
\end{definition}

A best-response of a player $i$ to an action profile will often be denoted by $\alpha_i^*$. Note that, according to the above definition, the best response of a player is never empty. We also introduce the following notion of a coordination game.
\begin{definition}[Coordination game]		\label{def:CoordinationGame}
A strategic-form game satisfying the positive-utility property (Property~\ref{P:PositiveUtilityProperty}) is a coordination game if, for every action profile $\alpha$ and player $i$, $u_j(\alpha_i',\alpha_{-i})\geq u_j(\alpha_i,\alpha_{-i})$ for any $\alpha_i'\in{\rm BR}_i(\alpha)$.
\end{definition}

In other words, a coordination game is such that at any action profile, if a player plays a best response, then no other player gets worse-off. This is satisfied by default when the current action profile is a Nash equilibrium, since a player's best response is to play the same action. 

In order to address stochastic stability, we will further need to introduce the notion of the best-BR (briefly, BBR). 

\begin{definition}[Best-BR]		\label{def:BestBR}
Let $i^*:\mathcal{A}\to\mathcal{I}$ be defined as:
\begin{equation*}
i^{*}(\alpha) \df \arg\max_{i\in\mathcal{I}}\left\{u_i(\alpha_i,\alpha_{-i}):\alpha_i\in{\rm BR}_i(\alpha)\right\}.
\end{equation*} 
The one-step transition $\alpha=(\alpha_{i^*},\alpha_{-i^*})\to(\alpha_{i^*}^*,\alpha_{-i^*})$, where $\alpha_{i^*}^*\in{\rm BR}_{i^*}(\alpha)$, is the best-BR to the current action profile $\alpha$ and will briefly be denoted by ${\rm BBR}(\alpha)$. 
\end{definition}

In other words, ${\rm BBR}(\alpha)$ corresponds to the one-step transition, where the player which changes its action receives the largest utility among all possible one-step transitions from $\alpha$.

\begin{lemma}	\label{Lm:CoordinationGameDeltaResistance}
Let $\mathcal{S}_{\rm NE}$ be the set of p.s.s.'s which correspond to the set of pure Nash equilibria. In any coordination game, the $\{\cS_{\rm NE}\}$-graph that attains the minimum $\delta$-resistance is: $g^*(\cS_{\rm NE}) = 
\left\{(s_k\to s_{\ell}): \alpha^{(\ell)} \in {\rm BBR}(\alpha)\right\}.$
\end{lemma}
\begin{proof}
Under the coordination property, and starting from any state $s\notin\cS_{\rm NE}$, we can construct a path starting from $s$ and leading to $\cS_{\rm NE}$ that consists only of one-step best-BR's. Such a path will include no cycles (since the utility of all players may not decrease along such path). Furthermore, such path of best-BR's may only terminate at a Nash equilibrium. 

By Definition~\ref{Df:W-graph} of a $\{\cS_{\rm NE}\}$-graph, a state $s\notin\cS_{\rm NE}$ is the source of exactly one arrow. Among the possible arrows with source $s$, the one that corresponds to a best-BR is the one with the minimum $\delta$-resistance (since it provides the maximum possible destination utility). We conclude that the $\{\cS_{\rm NE}\}$-graph(s) consisting only of best-BR's provide the minimum $\delta$-resistance. 
\end{proof}

Lemma~\ref{Lm:CoordinationGameDeltaResistance} shows that the $\{\cS_{\rm NE}\}$-graph of minimum $\delta$-resistance is the graph consisting of the one-step best-BR's starting from any non-Nash action profile. Using this property, we can show that the set of Nash equilibria are the stochastically stable states in any coordination game.

\begin{theorem}[Stochastic stability in coordination games]	\label{Th:StochasticStabilityCoordinationGames}
In any coordination game of Definition~\ref{def:CoordinationGame}, as $\epsilon\downarrow{0}$ and $\lambda\downarrow{0}$, the stochastically stable pure strategy states satisfy $\mathcal{S}^*\subseteq\mathcal{S}_{\rm NE}.$
\end{theorem}
\begin{proof}
It suffices to show that all p.s.s.'s outside $\mathcal{S}_{\rm NE}$ provide a $\delta$-resistance which is strictly higher than the $\delta$-resistance of any Nash equilibrium in $\mathcal{S}_{\rm NE}$ (as Theorem~\ref{Th:StochasticallyStableStatesMinimumResistance} dictates). 

Consider an action profile $\alpha$ which is not a Nash equilibrium and the corresponding p.s.s. $s$. Consider the part of the optimal $\{\mathcal{S}_{\rm NE}\}$-graph which leads to $s$, i.e., 
$$g^*(s|\mathcal{S}_{\rm NE}) \df \left\{ (s_{k}\to s_{\ell})\in g^*(\mathcal{S}_{\rm NE}): \mbox{$\exists$ path from $s_{\ell}$ to $s$}\right\}.$$
In other words, $g^*(s|\mathcal{S}_{\rm NE})$ corresponds to the part of the minimum-resistance graph $g^*(\mathcal{S}_{\rm NE})$ whose arrows lead to $s$. This graph might be empty if $s$ is not a recipient of any arrow in $g^*(\mathcal{S}_{\rm NE})$. For the remainder of the proof, define the graphs:
$g_1 \df g^*(\mathcal{S}_{\rm NE})\backslash g^*(s|\mathcal{S}_{\rm NE})$, 
$g_2 \df g^*(s)\backslash g^*(s|\mathcal{S}_{\rm NE}).$
Note that, $g^*(s|\mathcal{S}_{\rm NE}) \subset g^*(s)$, i.e., the graph that leads to $s$ through the minimum resistance graph of $\mathcal{S}_{\rm NE}$ is also part of the minimum resistance graph of $s$. By construction, we also have $g^*(s|\mathcal{S}_{\rm NE}) \subset g^*(\mathcal{S}_{\rm NE})$. Thus, the exact same arrows (i.e., the ones in $g^*(s|\cS_{\rm NE})$) are subtracted from $g^*(\cS_{\rm NE})$ and $g^*(s)$ to define the graphs $g_1$ and $g_2$, respectively.

By definition of the $\{\cS_{\rm NE}\}$-graphs, a node within the set $\{\cS_{\rm NE}\}$ cannot be the source of any arrow in $g_1$. Similarly, node $s$ may not be the source of any arrow in $g_2$. Since $\magn{\cS_{\rm NE}}\geq{1}$, and the fact that only a single arrow may stem from any given node, we conclude that $\magn{g_1} \leq \magn{g_2}$, i.e., $g_2$ contains at least as many arrows as $g_1$. 

Furthermore, by construction of graphs $g_1$ and $g_2$, there exists at least one node $s'\notin\cS_{\rm NE}$ with the following property: $(s'\to s'')\in g_1$ such that $\alpha''\in{\rm BBR}(\alpha')$, and $(s'\to s''')\in g_2$ such that $\alpha'''\notin {\rm BBR}(\alpha')$. This is due to the fact that any path in $g_2$ should eventually lead to $s\notin\cS_{\rm NE}$.

Thus, we conclude that $g_2$ contains at least as many arrows as $g_1$, and $g_2$ contains arrows which are not best-BR steps. Since only best-BR transition steps achieve the minimum resistance, we conclude that $\varphi(s|g_2) > \varphi(s|g_1)$, which implies that any $\{s\}$-graph may only have larger $\delta$-resistance as compared to the minimum $\delta$-resistance of $g^*(\cS_{\rm NE})$.
\end{proof}

\subsection{Simulation study in distributed network formation}		\label{sec:SimulationStudy}

In this section, we perform a simulation study of the proposed learning dynamics in a class of network formation games \cite{chasparis_network_2013}. We consider $n$ nodes deployed on the plane and assume that the set of actions of each node or agent $i$, $\cA_{i}$, contains all
possible combinations of neighbors of $i$, denoted ${N}_{i}$, with which a link can be established, i.e., $\cA_{i}=2^{{N}_{i}}$, including the empty set. Links are considered unidirectional, and a link established by node $i$ with node $j$, denoted $(j,i)$, starts at $j$ with the arrowhead pointing to $i$. A \emph{graph} $G$ is defined as a collection of nodes and directed links. Define also a \emph{path} from $j$ to $i$ as a sequence of nodes and directed links that starts at $j$ and ends to $i$ following the orientation of the graph, i.e., $$(j\rightarrow{i}) =
\bigl\{j=j_{0},(j_{0},j_{1}),j_{1},\dotsc,(j_{m-1},j_{m}),j_{m}=i\bigr\}$$ for some positive integer $m$. In a \emph{connected} graph, there is a path from any node to any other node.

Let us consider a utility function $u_{i}:\cA\rightarrow\mathbb{R}$, such that
\begin{equation}\label{E-NFG}
 u_{i}(\alpha) \df
 \sum_{j\in\cI\setminus\{i\}}\chi_{\alpha}(j\rightarrow{i})
 - \kappa \magn{\alpha_{i}},
\end{equation}
$i\in\cI$, where $\magn{\alpha_{i}}$ denotes the number of links corresponding to $\alpha_{i}$ and $\kappa$ is a constant in $(0,1)$. Also,
\begin{equation*}
 \chi_{\alpha}(j\to{i}) \df \begin{cases}
 1 & \mbox{if } (j\to{i})\subseteq G_{\alpha}\,,\\
 0 & \mbox{otherwise,}
 \end{cases}
\end{equation*}
where $G_\alpha$ denotes the graph induced by joint action $\alpha$. 
As it was shown in Proposition~4.2 in
\cite{chasparis_network_2013}, \emph{a network $G^*$ is a Nash equilibrium if and only if it is \emph{critically connected}, i.e., i) it is connected, and ii) for any $(s,i)\in{G}$, $(s\rightarrow i)$ is the unique path from $s$ to $i$}. For example, the Nash equilibria for $n=3$ agents and unconstrained neighborhoods are shown in \figurename~\ref{fig:NN}.

\begin{figure}[t!]
\centering
\includegraphics[scale=1]{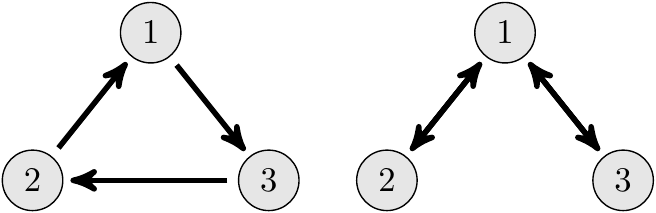}
\caption{Nash networks in case of $n=3$ agents and $0<\nu<1$.}
\label{fig:NN}
\end{figure}

\begin{figure*}[t!]
\centering
\begin{minipage}{0.49\textwidth}
\centering
\includegraphics[scale=1]{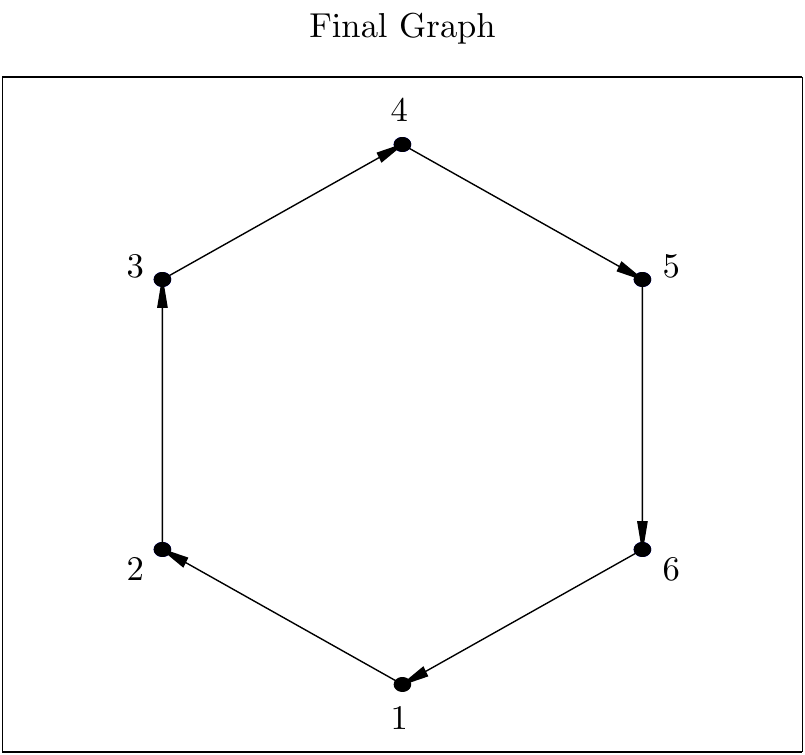}\\ (a)
\end{minipage}
\begin{minipage}{0.49\textwidth}
\centering
\vspace{27pt}
\includegraphics[scale=1]{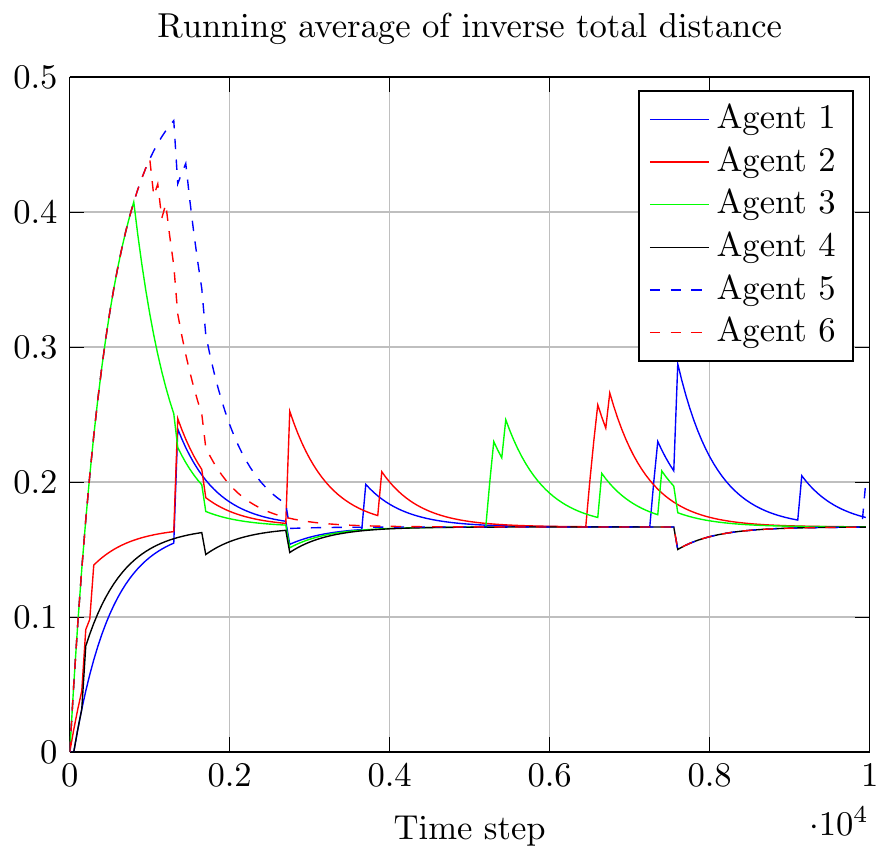} \\ (b)
\end{minipage}
\caption{(a) Final graph and (b) running-average inverse total distance with time under the perturbed learning automata dynamics of Table~\ref{Tb:ReinforcementLearning} when applied to the network formation game.}
\label{fig:network_formation_simulation}
\end{figure*}

\begin{proposition}\label{Pr:NetworkFormationCoordinationGame}
The network formation game defined by \eqref{E-NFG} is a coordination game. 
\end{proposition}
\begin{proof}
First, note that any network formation game with the utility of (\ref{E-NFG}) satisfies the positive-utility property. This is due to the fact that for any single link of cost $\kappa\in(0,1)$, an agent receives utility of at least $1$. For any joint action $\alpha\notin\cA^*$ assume  that a node $i$ picks its best response. Then no other agent becomes worse off, since a best response of any node $i$ always retains connectivity. Note that this is not necessarily true for any other change in actions. Thus, the coordination property of Definition~\ref{def:CoordinationGame} is satisfied. 
\end{proof}

\figurename~\ref{fig:network_formation_simulation} depicts the response of the learning dynamics in the network formation game. We consider $6$ nodes deployed on the plane, where the neighbors of each node are defined as the two immediate nodes (e.g., the neighbors of node $1$ are $N_1=\{2,6\}$).  
According to Theorem~\ref{Th:StochasticStabilityCoordinationGames}, in order for the average behavior to be observed, $\lambda$ and $\epsilon$ need to be sufficiently small. We choose: $\epsilon=\lambda=0.005$, and $\kappa=\nicefrac{1}{2}$. 

Given the large number of actions, we do not plot the strategy vector for each node. Instead, we plot the inverse total distance from each node to its neighboring nodes. In a wheel structure (and only under this structure), the inverse total distance to the neighboring nodes is equal to $\nicefrac{1}{1+5}=\nicefrac{1}{6}\approx 0.167$. The wheel structure is among the Nash equilibria of this game (as shown in \cite{chasparis_network_2013}) and the unique payoff-dominant equilibrium (i.e., every node receives its maximum utility). The wheel structure is the emergent structure as shown in \figurename~\ref{fig:network_formation_simulation}.

The simulation of \figurename~\ref{fig:network_formation_simulation} verifies Theorem~\ref{Th:StochasticStabilityCoordinationGames}, since convergence (in a weak sense) is attained to the set of Nash equilibria. However, it also demonstrates the potential of this class of dynamics for stronger convergence results, since the emergent Nash equilibrium is also payoff-dominant.

\section{Conclusions \& Future Work}    \label{sec:Conclusions}

In this paper, we considered a class of reinforcement-based learning dynamics that belongs to the family of discrete-time replicator dynamics and learning automata, and we provided an explicit characterization of the invariant probability measure of the induced Markov chain. Through this analysis, we demonstrated convergence (in a weak sense) to the set of pure strategy states, overcoming prior limitations of the ODE-method for stochastic approximations, such as the existence of a potential function. Furthermore, we provided a simplified methodology for computing the set of stochastically stable states, and we demonstrated its utility in the context of coordination games. This is the first result in this class of dynamics that demonstrates global convergence properties with no restrictions in the number of players and without requiring the existence of a potential function. Thus, it opens up new possibilities for the use of reinforcement-based learning in distributed control of multi-agent systems.

\appendices

\section{Proof of Proposition~\ref{Pr:WeakFeller}}	\label{Ap:WeakFeller}

Let us first consider the perturbed process $P_{\lambda}$. Let us also consider any sequence $\{z^{(k)}=(\alpha^{(k)},x^{(k)})\}$ such that $z^{(k)}\to{z}=(\alpha,x)\in\cZ$. For any open set $O\in\Bor(\cZ)$, 
\begin{eqnarray*}
\lefteqn{P_{\lambda}(z^{(k)}=(\alpha^{(k)},x^{(k)}),O) }\cr 
& = & \sum_{\alpha\in\mathcal{P}_{\cA}(O)}\Big\{\prod_{i=1}^{n}\tilde{x}_{i\alpha_i}^{(k)} \cdot \prod_{i=1}^{n}\Prob_{z^{(k)}}[\mathcal{R}_{i}(\alpha,x_i^{(k)})\in\mathcal{P}_{\cX_i}(O)] \Big\} \cr
& = & \sum_{\alpha\in\mathcal{P}_{\mathcal{A}}(O)} \Big\{\prod_{i=1}^{n}\mathbb{I}_{\mathcal{P}_{\cX_i}(O)}(\mathcal{R}_i(\alpha,x_i^{(k)})) \tilde{x}_{i\alpha_i}^{(k)}\Big\},
\end{eqnarray*}
where $\mathcal{P}_{\cX_i}(O)$ and $\mathcal{P}_{\mathcal{A}}(O)$ are the \emph{canonical projections} defined by the product topology, and 
$\tilde{x}_{i\alpha_i}^{(k)}\df (1-\lambda)x_{i\alpha_i}^{(k)} + {\lambda}/{\magn{\cA_i}}.$ Similarly, we have: 
\begin{equation*}
P_{\lambda}(z,O) = \sum_{\alpha\in\mathcal{P}_{\mathcal{A}}(O)} \Big\{\prod_{i=1}^{n}\mathbb{I}_{\mathcal{P}_{\cX_i}(O)}\left(\mathcal{R}_i\left(\alpha,x_i\right)\right) \tilde{x}_{i\alpha_i}\Big\}.
\end{equation*}

To investigate the limit of $P_{\lambda}(z^{(k)},O)$ as $k\to\infty$, it suffices to investigate the behavior of the sequence $\zeta_i^{(k)} \df \mathbb{I}_{\mathcal{P}_{\cX_i}(O)}(\mathcal{R}_i(\alpha,x_i^{(k)})).$ We distinguish the following (complementary) cases:

(a) $\mathcal{R}_i(\alpha,x_i)\notin\mathcal{P}_{\cX_i}(O)$ and $\mathcal{R}_i(\alpha,x_i)\notin\partial\mathcal{P}_{\cX_i}(O)$: In this case, there exists an open ball about the next strategy vector that does not share any common points with $\mathcal{P}_{\cX_i}(O)$. Due to the continuity of the function $\mathcal{R}_i(\alpha,\cdot)$, we have that $\zeta_i^{(k)}\to \zeta_i \df \mathbb{I}_{\mathcal{P}_{\cX_i}(O)}(\mathcal{R}_i(\alpha,x_i))\equiv{0}$.

(b) $\mathcal{R}_i(\alpha,x_i)\in\mathcal{P}_{\cX_i}(O)$: In this case, there exists an open ball about the next strategy vector that belongs to $\mathcal{P}_{\cX_i}(O)$, since $O\in\Bor(\cZ)$. Due to the continuity of the function $\mathcal{R}_i(\alpha,\cdot)$, we have that $\zeta_i^{(k)}\to \zeta_i=1$.

(c) $\mathcal{R}_i(\alpha,x_i)\notin\mathcal{P}_{\cX_i}(O)$ and $\mathcal{R}_i(\alpha,x_i)\in\partial\mathcal{P}_{\cX_i}(O)$: In this case,  $\zeta_i \equiv{0}$. We conclude that $\liminf_{k\to\infty}{\zeta_i^{(k)}}\geq \zeta_i = 0$, since $\zeta_i^{(k)}\in\{0,1\}$.

In either one of the above (complementary) cases, we have that $\liminf_{k\to\infty}{\zeta_i^{(k)}}\geq \zeta_i$. Finally, due to the continuity of the perturbed strategy vector $\tilde{x}_{i\alpha_i}$ with respect to $x_{i\alpha_i}$, we conclude that for any sequence $z^{(k)}\to z$, $\liminf_{k\to\infty}P_{\lambda}(z^{(k)},O) \geq P_{\lambda}(z,O).$ Thus, by \cite[Proposition~7.2.1]{Lerma03}, $P_{\lambda}$ satisfies the weak Feller property.

The above derivation can be generalized to any selection probability function $f(x_{i\alpha_i})$ in the place of $\tilde{x}_{i\alpha_i}$, provided that it is a continuous function. Thus, the proof for the unperturbed process $P$ follows the exact same reasoning by simply setting $f(x_{i\alpha_i})=x_{i\alpha_i}$. 

\section{Proof of Proposition~\ref{Pr:ConvergenceToPSS}}	\label{Ap:ConvergenceToPSS}

(a) Let us consider an action profile $\alpha=(\alpha_1,...,\alpha_n)\in\mathcal{A}$, and an initial strategy profile $x(0)=(x_1(0),...,x_n(0))$ such that $x_{i\alpha_i}(0)>0$ for all $i\in\mathcal{I}$. Note that if the same action profile $\alpha$ is selected consecutively up to time $t$, then the strategy of agent $i$ satisfies:
\begin{equation}	\label{eq:ConvergenceToPSS:AccumulatedStrategy}
x_{i}(t) = e_{\alpha_i} - (1-\epsilon u_i(\alpha))^{t}(e_{\alpha_i}-x_i(0)).
\end{equation}
Given that $B_t$ is non-increasing, from continuity from above we have
\begin{equation}	\label{eq:ConvergenceToPSS:Binfty}
\Prob_{z}[B_{\infty}] = \lim_{t\to\infty}\Prob_{z}[B_t] = \lim_{t\to\infty}\prod_{k=0}^{t}\prod_{i=1}^{n}x_{i\alpha_i}(k).
\end{equation}
Note that $\Prob_{z}[B_{\infty}] > 0$ if and only if 
\begin{equation}	\label{eq:ConvergenceToPSS:Condition1}
\sum_{t=0}^{\infty}\log(x_{i\alpha_i}(t)) > -\infty, \mbox{ for all } i\in\mathcal{I}.
\end{equation}
Let us introduce the variable $y_i(t) \df 1-x_{i\alpha_i}(t),$ which corresponds to the probability of agent $i$ selecting any action other than $\alpha_i$. Condition (\ref{eq:ConvergenceToPSS:Condition1}) is equivalent to
\begin{equation}	\label{eq:ConvergenceToPSS:Condition2}
-\sum_{t=0}^{\infty}\log(1-y_i(t)) < \infty,	\mbox{ for all } i\in\mathcal{I}.
\end{equation}
Note that $y_{i}(t+1)/y_i(t) = 1-\epsilon u_i(\alpha) < 1$, which (by the Ratio test, cf.,~\cite[Theorem~6.2.4]{Reed98}) implies that the series of positive terms $\sum_{t=1}^{\infty}y_i(t)$ is convergent. Hence, $\lim_{t\to\infty}y_i(t) = 0$. Thus, from L'Hospital's rule (cf.,~\cite[Theorem~5.13]{Rudin64}),
\begin{equation}	\label{eq:LHopital}
\lim_{t\to\infty}\frac{-\log(1-y_i(t))}{y_i(t)} = \lim_{t\to\infty}\frac{1}{1-y_i(t)} = 1 > 0. 
\end{equation}
From the Limit Comparison Test (cf.,~\cite[Theorem~6.2.2]{Reed98}), we conclude that condition (\ref{eq:ConvergenceToPSS:Condition2}) holds, which equivalently implies that $\Prob_z[B_{\infty}]>0$. 
Lastly, due to (\ref{eq:ConvergenceToPSS:AccumulatedStrategy}), $\Prob_z[B_{\infty}]$ is continuous with respect to $x(0)$ which takes values in a bounded and closed set $\cX$. Thus, by \cite[Theorem~3.2.2]{Reed98}, we conclude that $\inf_{z\in\cZ}\Prob_{z}[B_{\infty}] > 0$.

(b) Define the set
$C_\ell \df \left\{z\in\cZ:|x_i|_{\infty} > 1 - \epsilon^{\ell}\,, \forall i\in\mathcal{I}\right\},$
where $|x_i|_{\infty}\df\max\{x_{i\alpha_i}, \alpha_i\in\mathcal{A}_i\}$, i.e., $C_{\ell}$ corresponds to a strategy being close to a vertex of $\cX$. For $\ell > 0$,
\begin{eqnarray}	\label{eq:ConvergenceToPss:At1}
\Prob_{z}[A_t] & \geq & \sum_{k=1}^{t}\Prob_z[\uptau(C_{\ell})=k\,, Z\circ\theta_k\in B_{\infty}] \cr
& = & \sum_{k=1}^{t}\Prob_{z}[Z\circ\theta_k\in B_{\infty} | \uptau(C_{\ell})=k]\cdot \Prob_{z}[\uptau(C_{\ell})=k] \cr
& \geq & \inf_{z\in C_{\ell}} \Prob_{z}[B_{\infty}] \cdot \sum_{k=1}^{t} \Prob_{z}[\uptau(C_{\ell})=k] \cr 
& \geq & \inf_{z\in C_{\ell}}\Prob_{z}[B_{\infty}]\cdot \inf_{z\in C_\ell^{c}} \Prob_{z}[\uptau(C_{\ell})\leq{t}],
\end{eqnarray}
where the second inequality is due to the Markov property. Consider the subsequence $t_k = k \ell^{m}$, for some $m=m(\ell)>{0}$ such that, the time block of $\ell^{m}$ iterations is sufficiently large so that $C_{\ell}$ can be reachable from any state in $C_{\ell}^{c}$. Then,
\begin{equation*}	
\Prob_{z}[\uptau(C_{\ell})\leq{t_k} | \uptau(C_{\ell})>{t_{k-1}}] \geq \inf_{z\in C_{\ell}^{c}}\Prob_{z}[B_{\ell^m}] \geq \inf_{z\in C_{\ell}^{c}}\Prob_{z}[B_{\infty}],
\end{equation*}
where the last inequality is due to (\ref{eq:ConvergenceToPSS:Binfty}). Given (a), and for any $\ell>0$, the r.h.s. of the above inequality is bounded away from zero. Hence, from the counterpart of the Borel-Cantelli Lemma (cf.,~\cite[Section~3.3]{Breiman92}) and the fact that $\{\uptau(C_{\ell})\leq{t_k}\} \subseteq \{\uptau(C_{\ell})\leq{t_{k+1}}\}$, we have that, for any $\ell > 0$,
\begin{equation}	\label{eq:ConvergenceToPss:At2}
\lim_{k\to\infty}\inf_{z\in C_{\ell}^c}\Prob_{z}[\uptau(C_{\ell})\leq{t_k}] = 1.
\end{equation}
Finally, set $k=\ell$. Then, $t_k=t_\ell=\ell^{m+1}$. Given (\ref{eq:ConvergenceToPss:At1})--(\ref{eq:ConvergenceToPss:At2}) and from continuity from below, we have
\begin{equation*}
\Prob_{z}[A_{\infty}] = \lim_{\ell\to\infty}\Prob[A_{t_\ell}] \geq \lim_{\ell\to\infty} \inf_{z\in C_{\ell}}\Prob_{z}[B_{\infty}] = 1,
\end{equation*}
where the last equality is due to the definition of $C_\ell$ and (\ref{eq:ConvergenceToPSS:Binfty}).

\section{Proof of Lemma~\ref{Lm:StationaryDistributionApproximation}} \label{Ap:StationaryDistributionApproximation}

(a) The state $z'$, realized after agent $j$ trembled and played $\alpha_j'$ starting from $s$, is uniquely defined as $z'\df(\alpha',e_{\alpha_j}+\epsilon u_j(\alpha')(e_{\alpha_j'}-e_{\alpha_j})).$ Thus, we can write: 
\begin{eqnarray*}
\lefteqn{QP^{t}(s,\Neig_{\delta}(s'))}\cr & = & \int_{\cZ}\gamma_j\Dirac{z'}(dy)P^{t}(y,\Neig_{\delta}(s'))=\gamma_jP^{t}(z',\Neig_{\delta}(s')).
\end{eqnarray*}
Given that $\Neig_{\delta}(s')$ is a continuity set of $Q\Pi(s,\cdot)$, from Portmanteau theorem we have that, for any $\delta>0$, $$\hat{P}_{ss'}=Q\Pi(s,\Neig_{\delta}(s')) = \gamma_j \lim_{t\to\infty}P^t(z',\Neig_{\delta}(s')).$$ Note also that $P^{t}(z',\Neig_{\delta}(s'))\leq \Prob_{z'}[\uptau(\Neig_{\delta}(s'))\leq{t}]$. Since the sequence of events $\{\uptau(\Neig_{\delta}(s'))\leq{t}\}_t$ is non-decreasing, then from continuity from below, we have that, for any $\delta>0$, 
\begin{equation}	\label{eq:StationaryDistributionApproximation:EqA1}
\lim_{t\to\infty}P^{t}(z',\Neig_{\delta}(s')) \leq \Prob_{z'}[\uptau(\Neig_{\delta}(s'))\leq\infty].
\end{equation}
On the other hand, we have
\begin{eqnarray*}
P^{t}(z',\Neig_{\delta}(s')) & \geq & \sum_{k=1}^{t}\Prob_{z'}[\uptau(\Neig_{\delta}(s'))=k,Z\circ\theta_k\in{B}_{\infty}] \cr
&\geq & \inf_{z\in\Neig_{\delta}(s')}\Prob_{z}[B_\infty] \cdot \Prob_{z'}[\uptau(\Neig_{\delta}(s'))\leq{t}],
\end{eqnarray*}
where in the second inequality we have used the Markov property.
Given that $\lim_{\delta\downarrow{0}}\inf_{z\in\Neig_{\delta}(s')}\Prob_{z}[B_\infty]=1,$ we get 
\begin{equation}	\label{eq:StationaryDistributionApproximation:EqA2}
\lim_{\delta\downarrow{0}}\lim_{t\to\infty}P^{t}(z',\Neig_{\delta}(s')) \geq \lim_{\delta\downarrow{0}} \Prob_{z'}[\uptau(\Neig_{\delta}(s'))\leq\infty].
\end{equation}
The conclusion follows directly from (\ref{eq:StationaryDistributionApproximation:EqA1})--(\ref{eq:StationaryDistributionApproximation:EqA2}).

(b) Consider the unperturbed process initiated at state $z'$, i.e., $Z_0=z'$. Let us also define the set $$D_{j,\ell}(\alpha') \df \left\{(\alpha,x)\in\cZ: x_{j\alpha_j'} > 1 - H_j(\alpha')^{\ell}\right\},$$ where $H_j(\alpha')\df 1-\epsilon u_j(\alpha')$. The set $D_{j,\ell}(\alpha')$ is the unreachable set in the strategy space of agent $j$ when starting from $x_{j\alpha_j'}=0$ and playing action $\alpha_j'$ for $\ell$ consecutive steps. Define also the set $$E_{j,\ell}(\alpha') \df D_{j,\ell+1}(\alpha')^c\cap D_{j,\ell}(\alpha').$$

One possibility for realizing a transition from $z'$ to $\Neig_{\delta}(s')$ is to follow the shortest path, that is, the path of playing action $\alpha'$ consecutively. Thus,
\begin{equation} \label{eq:LowerBoundOneStepTransition}
\breve{P}_{ss'}(\delta)
\geq \Prob_{z'}\left[\alpha(t+1)=\alpha', \forall t < \uptau_{s}^*(\Neig_{\delta}(s'))\right].
\end{equation}

When the process reaches $\Neig_{\delta}(s')$ for the first time, action profile $\alpha'$ has been played for at least $\uptau_{s}^*(\Neig_{\delta}(s'))$ times.\footnote{Let us assume that along a sample path from $z'$ to $\Neig_{\delta}(s')$ and at iteration $t_k$, the strategy of agent $j$ with respect to action $\alpha_j'$ is $x_{j\alpha_j'}(t_k)=\rho>0$. If agent $j$ selects action $\alpha_j'$ at time $t_k+1$, it's next strategy will be:
$$x_{j\alpha_j'}(t_k+1) = \rho + \epsilon u_j(\alpha')(1-\rho) = \epsilon u_j(\alpha') + H_j(\alpha')\rho \df x_{j\alpha_j'}^*.$$ 
If, instead, agent $j$ selects action $\alpha_j\neq\alpha_j'$ at time $t_k+1$ and then $\alpha_j'$ at time $t_k+2$, i.e., it deviates from playing action $\alpha_j'$, then the strategy evolves as:
\begin{eqnarray*}
x_{j\alpha_j'}(t_k+1) & = & \rho + \epsilon u_j(\alpha) (-\rho) \cr & = & H_j(\alpha)\rho, \cr
x_{j\alpha_j'}(t_k+2) & = & H_j(\alpha)\rho + \epsilon u_j(\alpha') (1-H_j(\alpha)\rho) \cr 
& = & (H_j(\alpha')\rho)H_j(\alpha) + \epsilon u_j(\alpha') \cr
& < & x_{j\alpha_j'}^*,
\end{eqnarray*}
since $0<\epsilon u_j(\alpha) < 1$. Informally, any single deviation from the shortest path to $s'$ cannot recover the drop in the strategy at the next iteration. Thus, along any path from $z'$ to $\Neig_{\delta}(s')$, when the process reaches $\Neig_{\delta}(s')$ for the first time, action $\alpha'$ has been played for at least $\uptau_{s}^*(\Neig_{\delta}(s'))$ times, which is the number of iterations required for reaching $\Neig_{\delta}(s')$ along the shortest path.} 
Define an iteration subsequence $\{t_k\}$, with $Z_{t_0}=Z_0=z'$, such that, at any time $t_k$ action $\alpha'$ is selected for the next iteration (i.e., $\alpha(t_0+1)=...=\alpha(t_k+1)=\alpha'$). Due to the Markov property, in order for the unperturbed process to reach $\Neig_{\delta}(s')$, there exists time $t_k$ such that $\alpha(t_k+1)=\alpha'$ while $Z_{t_k}\in E_{j,k}(\alpha')$.
Thus, 
\begin{eqnarray*}
\breve{P}_{ss'}(\delta)  
& \leq & \Prob_{z'}\left[\exists\{t_k\}:\alpha(t_{k}+1)=\alpha', Z_{t_k}\in E_{j,k}(\alpha'),\right. \cr && \left. \mbox{ for all } k < \uptau_s^*(\Neig_{\delta}(s')) \right]. 
\end{eqnarray*}
Using the properties of the conditional probability, we also have:
\begin{eqnarray*}
\breve{P}_{ss'}(\delta) & \leq & \Prob_{z'}\left[\exists\{t_k\}:\alpha(t_k+1)=\alpha'\,, \forall k < \uptau_s^*(\Neig_{\delta}(s')) |  \right. \cr && \left. Z_{t_k}\in E_{j,k}(\alpha')\right]. 
\end{eqnarray*}
Using again the Markov property,  
\begin{eqnarray*}
\lefteqn{\breve{P}_{ss'}(\delta)}\cr & \leq & 
\prod_{t< \uptau_s^*(\Neig_{\delta}(s'))} \sup_{z\in E_{j,t}(\alpha')}\Prob_{z'}\left[\alpha(t+1)=\alpha'|Z_{t}=z\right] \cr
& = & \Prob_{z'}\left[\alpha(t+1)=\alpha', t < \uptau_s^*(\Neig_{\delta}(s'))\right].
\end{eqnarray*}
Given also (\ref{eq:LowerBoundOneStepTransition}), the conclusion follows.

(c) The minimum first hitting time to the set $\Neig_{\delta}(s')$ satisfies: $$\uptau_s^*(\Neig_{\delta}(s'))=\left\lceil{\frac{\log(\delta)}{\log(H_j(\alpha'))}}\right\rceil \df T(\epsilon),$$ where $H_j(\alpha')\df 1-\epsilon u_j(\alpha')$. There exists correction factor $c=c(\epsilon,\delta)\in[0,1)$, such that 
$$T(\epsilon)=\frac{\log(\delta)}{\log(H_j(\alpha'))} + c(\epsilon,\delta).$$
Due to statement (b), and for sufficiently small $\epsilon>0$, we have:
\begin{equation}	\label{eq:Lemma5.3-0}
\log\left(\breve{P}_{ss'}(\delta)\right) \approx \sum_{t=1}^{T(\epsilon)}\log\left(1-H_j(\alpha')^{t}\right).
\end{equation}
In the remainder of the proof, we will approximate the r.h.s. of (\ref{eq:Lemma5.3-0}). 

To simplify notation, denote $H\df H_j(\alpha')$. Note that
\begin{eqnarray*}	
\lefteqn{\lim_{\epsilon\downarrow{0}}\log\left(H^{T(\epsilon)}\right)} \cr
&=& \lim_{\epsilon\downarrow{0}}\left\{\left(\frac{\log(\delta)}{\log(H)} + c(\epsilon,\delta)\right)\log(H)\right\} = \log(\delta),
\end{eqnarray*}
and due to the continuity of the natural logarithm, 
\begin{equation}	\label{eq:Lemma5.3-1}
\lim_{\epsilon\downarrow{0}}H^{T(\epsilon)}=\delta.
\end{equation}
As a result, for any $\ell\in\mathbb{N}$, $$\lim_{\epsilon\downarrow{0}}H^{\ell T(\epsilon)} = \delta^{\ell}.$$ By Taylor series expansion of the natural logarithm (for small argument values), we have:
\begin{equation*}
\log\left(1-H^{t}\right) \approx -\sum_{\ell=1}^{\infty}\frac{H^{\ell t}}{\ell}.
\end{equation*}
Thus,
\begin{eqnarray*}
\lefteqn{(1-H)\sum_{t=1}^{T(\epsilon)}\log\left(1-H^{t}\right)}\cr & \approx & -\sum_{\ell=1}^{\infty}\frac{1}{\ell}\Big[(1-H)\sum_{t=1}^{T(\epsilon)}H^{\ell t}\Big] \cr 
& = & -\sum_{\ell=1}^{\infty}\frac{1}{\ell}\Big[(1-H)\frac{1-H^{\ell(T(\epsilon)+1)}}{1-H^\ell}-(1-H)\Big] \cr
& = & -\sum_{\ell=1}^{\infty}\frac{1}{\ell}\Big[\frac{1-H^{\ell(T(\epsilon)+1)}}{1+H+\cdots + H^{\ell-1}} -(1-H)\Big]
\end{eqnarray*}
Note that, for any $\ell\in\mathbb{N}$, $H^{\ell}\to{1}$ as $\epsilon\downarrow{0}$. Thus, we have
\begin{eqnarray*}
\lefteqn{\lim_{\epsilon\downarrow{0}}(1-H)\sum_{t=1}^{T(\epsilon)}\log\left(1-H^{t}\right)}\cr & \approx & -\sum_{\ell=1}^{\infty}\frac{1}{\ell^2}(1-\delta^{\ell})\df\eta(\delta),
\end{eqnarray*}
which corresponds to a negative finite constant. Hence, using the fact that $1-H=\epsilon u_j(\alpha')$, we conclude that, for sufficiently small $\epsilon>0$,
$$\log(\breve{P}_{ss'}(\delta)) \approx \frac{\eta(\delta)}{\epsilon u_j(\alpha')}.$$


\bibliographystyle{IEEEtran}
\bibliography{Bibliography.bib}

%

\end{document}